\documentclass[pdftex,11pt]{article}%
\usepackage[hidelinks]{hyperref}%
\usepackage[pdftex]{graphicx}%
\usepackage{threeparttable}
\usepackage[normalem]{ulem}
\usepackage{amsmath,amsthm,amsfonts,
amsbsy,amssymb,upref,enumerate,bigstrut,
color,mathtools,mathrsfs,float,bm,dsfont,scalefnt,float,subfloat,bbm}
\usepackage{lipsum}
\usepackage{stmaryrd}
\usepackage{authblk} 
\usepackage[T1]{fontenc}
\usepackage{lmodern}
\usepackage{subcaption}
\usepackage{multirow}
\usepackage[toc,page]{appendix}

\def\orcid#1{\kern .08em\href{https://orcid.org/#1}{\includegraphics[keepaspectratio,width=0.7em]{parametros/orcid.pdf}}}

\theoremstyle{definition}

\newtheorem{theorem}{Theorem}[section]

\newtheorem{corollary}[theorem]{Corollary}

\newtheorem{lemma}[theorem]{Lemma}
\newtheorem{proposition}[theorem]{Proposition}
\newtheorem{remark}[theorem]{Remark}
\numberwithin{equation}{section} 

\makeatletter
\def\@seccntformat#1{\@ifundefined{#1@cntformat}%
	{\csname the#1\endcsname\quad}
	{\csname #1@cntformat\endcsname}
}
\makeatother

\newif\ifShowComments
\ShowCommentstrue
\def\strutdepth{\dp\strutbox}
\def\druk#1{\strut\vadjust{\kern-\strutdepth
        {\vtop to \strutdepth{%
                \baselineskip\strutdepth\vss
                        \llap{\hbox{#1}\quad}\null}}}}

\title{\bf A General Class of Trimodal Distributions: Properties and Inference}

\author[1]{Roberto Vila \thanks{rovig161@gmail.com}}
\author[1]{Victor Serra  \thanks{victorserra92@gmail.com} }
\author[2,3]{Mehmet N. Çankaya \thanks{mehmet.cankaya@usak.edu.tr}}
\author[4]{Felipe Quintino \thanks{felipe.quintino@unir.br}}
\affil[1]{Department of Statistics, University of Bras\'ilia, Bras\'ilia, Brazil}
\affil[2]{ Department of International Trading and Finance, Faculty of Applied Sciences, Uşak University, Uşak, Turkey;
	}
\affil[3]{ 
	Department of Statistics, Faculty of Art and Sciences, Uşak University, Uşak, Turkey
}
\affil[4]{Department of Mathematics and Statistics, Federal University of Rondônia,
Paraná, Brazil}
\date{\today}                     
\begin{document}
\maketitle

\begin{abstract}
The modality is important topic for modelling. Using parametric models is an efficient way when real data set shows trimodality. In this paper we propose a new class of trimodal probability distributions, that is, probability distributions that have up to three modes. Trimodality itself is achieved by applying a proper transformation to density function of certain continuous probability distributions. At first, we obtain preliminary results for an arbitratry density function $g(x)$ and, next, we focus on the Gaussian case, studying  trimodal Gaussian model more deeply. The Gaussian distribution is applied to produce the trimodal form of Gaussian known as normal distribution. The tractability of analytical expression of normal distribution, and properties of the trimodal normal distribution are important reasons why we choose normal distribution. Furthermore, the existing distributions should be improved to be capable of  modelling efficiently when there exists a trimodal form in a data set. After new density function is proposed, estimating its parameters is important. Since Mathematica 12.0 software has optimization tools and important modelling techniques,  computational steps are performed by using this software. The bootstrapped form of real data sets are applied to show the modelling ability of the proposed distribution when real data sets show trimodality. 
\end{abstract}


\section{Introduction}\label{introsec}

The modality is an important topic when the nature of phenomena can be modelled by using the function which can be capable to have different forms of peaks.  The modality can occur when there exists an irregularity in the output of an experiment. In the statistical view point, the random variables are nonidentically distributed. In other words, there can be a mixing of some populations even if same experiment is conducted while getting outputs of the corresponding experiment \cite{CBDA15,mixeddistr,denizspain,VHR21}. In such a situation, location and scale parameters of the mixed populations are important to get the central tendency and disperson (statistics) of the mixed form of the population in the experiment after shape, scale and bimodality parameters of function in the statistical theory are necessary components of a function which is used for conducting an efficient modelling  \cite{CanEEq,Dogq,mixeddistr,VFSPO20,VC20,popovicazza,Vilaetalbi01}. 

There are many well-known distributions used at the statistical inference in which regression and its counter parts such as time series, design of experiments, structural equation modelling in the applications from social science. There are different techniques to produce a probability density function \cite{felixgen}. The analytical tractability and properties of the proposed distribution are important when the new distribution is used for modelling. For example, the existence of moments and entropy function are important when the modelling on the real data sets are performed \cite{Lehmann1998,mcoa20}. The well-known normal distribution which is also known as Gaussian distribution should be transformed  into a trimodal form when data sets show trimodality. The advantage of using a trimodal
distribution is that a real data set can be a combination of two, three or more normal distributions with different parameter values for location and scale. The apperance of such mixed distributions can be in a trimodal form.  Especially, since the working principle of a phenomena depends on many factors, it is reasonable to expect that a trimodal form for a real data set can occur. In other words, it is assumed that the random variables are identically distributed. However, the identicality is a restrictive assumptation for modelling a data set. The parametric models are necessary to perform an efficient modelling when trimodal representation
in a data set exists due to the structure of non-identicality (hetero data) or the mixing form of distributions. On the other side, if observations $x_1,x_2, \cdots, x_n$ are distributed as a parametric model such as trimodal form, then they are identical, that is, it depends on where and how you look at the results of experiments, because we do already have finite sample size. While managing an efficient modelling on real data sets, it is reasonable to consult parametric models which can be capable of dealing with different forms of modality. The bimodality parameters  $\rho$ and $\delta$ with shape parameter $\alpha$  of Maxwell  distribution in  \cite{maxwelldist} are tools for us to generate different forms of modality. There are different degrees of mixing via Maxwell distribution and the expression which can help us to have different forms of modality of function for generating modality via compounding distributions  \cite{VC20,popovicazza}. The bimodal form on the positive part of the real line produces a trimodal form when it is reflected to the negative part of the real line via mirror imaging \cite{CBDA15,canent,Domma17,transmut}. The different degree of bimodality, i.e. the length of  periodicity of modality,  on the positive part of the real line can be constructed by producing a new objective function based on the deformation \cite{VC20,Can21frac}, because deformation which can be regarded as a kind of rescaling can make the different length for periodicity \cite{CanEEq,BercherFisq,Bercherq}.

The smooth kernel distribution will be used to fit data sets; because, the strict and soft forms of trimodal normal distribution should be compared with smooth kernel technique to perform a comparison among them.  
Three normal distributions, represented by functions $g(x; \mu_1, \sigma_1)$, $g(x; \mu_2, \sigma_2)$ and $g(x; \mu_3, \sigma_3)$, can be mixed to get the trimodal form of normal distribution $g$. Thus, the modelling performance of the trimodal distribution constructed for the normal distribution $g$ can be tested for the mixed data sets. In fact, the mixed form of two or three function with one mode and symmetric can show the symmetric form with trimodal representation.  In other words, the groups around location parameter can be divided into two forms \cite{CBDA15,mixeddistr,denizspain,VHR21}. 

The main aim in this paper is to propose a distribution with trimodal form on the real line via using a technique, as is given by \cite{VC20,denizspain,VHR21,bimodal}. We keep to follow the symmetric case and our aim is to focus on the trimodal form on the real line. When there exists a trimodal form in a data set, the location and scale parameters should be estimated efficiently. In other words, each group coming from groups $g(x; \mu_1, \sigma_1)$, $g(x; \mu_2, \sigma_2)$ and $g(x; \mu_3, \sigma_3)$ has its values for location and scale parameters. In our case, we try to estimate one location parameter and one scale parameter when the mixed data sets for two or three groups are used. It is important to note that the apperance of trimodality can occur via mixing the different values of parameters of functions. Since the true model for the mixed three normal distributions is chosen to estimate the location and scale parameters precisely, the performance of modelling will be increased greatly when
the probability density functions having modality property are taken into account. For example, the mixed distribution has  parameters which are the mixing proportion $w_1, w_2$ and $w_3=1-w_1-w_2$ for three groups which are necessary to estimate. For the mixed normal distributions, we have $\mu_1, \sigma_1, \mu_2, \sigma_2, \mu_3, \sigma_3, w_1$ and $w_2$ (see Section \ref{Sect:5}). In totally, there are 8 parameters which have to be estimated. However, in our case we have three main parameters and also $\mu$ and $\sigma$ of the distribution. Thus, there will be 5 parameters which will  be estimated. The optimization of the log-likelihood function according to these 5 parameters of a function can be easy to reach the global point of the $ \log(f) $ when compared with a function including 8 parameters. In addition, we have only one location and scale parameter which can be free from the mixing proportion  $w_1, w_2$ and $w_3=1-w_1-w_2$. The numerical computation while conducting the optimization of log-likelihood according to parameters can include the less numerical errors. Note that the numerical errors in the function with 8 parameters can be bigger than that of 5 parameters. Thus, the more precise evaluation can be achieved for the numerically precise evaluation of estimating the parameters $\mu$ and $\sigma$, which is why we prefer to consider proposing such a trimodality while conducting an efficient fitting on the real data sets. On the other side of the modelling perspective, the structure of grouping cannot be determined only for the used mixing proportions $w_1, w_2$ and $w_3=1-w_1-w_2$. There can be irrational proportions such as 1/6, 1/9, etc. for the mixing in a data set. The precise evaluation in computation  for the true value of irrational proportion cannot be performed accurately, which makes a disadvantage for us when we use the estimated values of these kinds of the proportions such as irrational ones. 

This paper is organized as follows. In Section \ref{A class of CDF}, we define the class of trimodal probabilistic models. In Section \ref{Structural properties}, some structural properties of the proposed model are examined.  We provide a formal proof for the trimodality of a class of symmetric kernel densities, present a stochastic representation, and we provide  the closed formulas for the moments, entropies and stochastic representation. The existence of these expressions is important to use the proposed distrbution for fitting the data sets. Section \ref{Sect:4} is divided to describe the proposed model when the normal (Gaussian) is applied (see Table \ref{table:1}). For this case, some properties such as modality, moments, Shannon entropy among others also are discussed. In Sections \ref{Sect:5}, we introduce the mixing form. Section  \ref{Sect:6}, represents the $\log_q$ likelihood function  used as a method for parameter estimation. In Section  \ref{appsect}, the real data sets are applied. Section  \ref{concsect} is for conclusion and future works. In Appendices as a supplementary material, we provide proofs and codes of Mathematica software.

\section{A class of  continuous probability distributions}
\label{A class of CDF}
\noindent
Let $g:D\subset \mathbbm{R}\to [0,\infty)$, $D={\rm supp}(g)\neq\emptyset$, be a kernel density  with corresponding cumulative distribution function (CDF) denoted by $G$. 
The function $g$ can be associated (or not) with an additional parameter $\xi$ (or vector $\boldsymbol{\xi}$).
For a random variable $X$ we define the following probability density function (PDF)
\begin{eqnarray}\label{pdfbgev2}
	f(x;\boldsymbol{\theta})
	=
	{\sigma \over Z_{\boldsymbol{\theta}}}\, 
	\biggl[\rho+\delta T\biggl({x-\mu\over\sigma};\alpha,p\biggr)\biggr] \,
g\biggl({x-\mu\over\sigma}\biggr),  \quad {x-\mu\over\sigma}\in  D,
\end{eqnarray}
where  $\boldsymbol{\theta}=(\mu,\sigma,\alpha,\rho, \delta)$ is a parameter vector such that $\mu\in\mathbbm{R}$ is a location parameter, $\sigma>0$ is a scale parameter, $\alpha>0$ is a shape parameter and $\rho\geqslant 0$  and $\delta\geqslant 0$ are parameters that control different forms of modality of distribution. Note that $\rho$ and $\delta$ cannot be zero simultaneously.
The function $Z_{\boldsymbol{\theta}}$ appearing in the definition of $f$ is a normalizing factor and the function $T:D\subset \mathbbm{R}\to (0,1)$ is given by
\begin{align}\label{tranf-G} 
T(x;\alpha;p)=
{\gamma(p,x^2/\alpha^{2})\over\Gamma(p)} \quad \text{with} \ p>0 \ \text{known}.
\end{align}
Here 
$\gamma(p,u)=\int_{0}^{u}w^{p-1}{\rm e}^{-w} {\rm d}w$ is
the incomplete gamma function and $\Gamma(p)$ is the gamma
function. For $p=3/2$ and $D=(0,\infty)$, the function $T(x;\alpha,p)$ on $D$ defines the Maxwell distribution with scale parameter $\alpha$.

Hereafter, we will denote $X\sim  {\rm TD}(\boldsymbol{\theta})$ for a random variable $X$ that follows the trimodal distribution \eqref{pdfbgev2}.

When $\delta= 0$ and $\rho\neq 0$ fixed in \eqref{pdfbgev2},  the original density $g$ is recovered.

Since $0<T(x;\alpha,p)<1$ for almost all $x\in D$, we have 
$0<Z_{\boldsymbol{\theta}}
\leqslant
(\rho+\delta)\sigma$.
A simple calculation shows that (see Corollary \ref{prop-gen})
\begin{align}\label{normalizator}
Z_{\boldsymbol{\theta}}
&=
(\rho+ \delta)\sigma+
\delta\sigma 
\left\{
{\mathbbm{E}}
\big[
G(-\alpha \sqrt{Y})
\big]
-
{\mathbbm{E}}
\big[
G(\alpha \sqrt{Y})
\big]
\right\},
\end{align}
where $Y\sim{\rm Gamma}(p,1)$ and $G$ is the corresponding CDF of $g$. Furthermore, the CDF of $X\sim {\rm TD}(\boldsymbol{\theta})$, denoted by $F(x;\boldsymbol{\theta})$, is written as
\begin{align} \label{CDF}
F(x;\boldsymbol{\theta})
&=
{\rho\sigma\over Z_{\boldsymbol{\theta}}}\, G\biggl({x-\mu\over\sigma}\biggr)
+
{\delta\sigma\over Z_{\boldsymbol{\theta}}}
\left\{	\mathbbm{E}\big[G(-\alpha\sqrt{Y})\big]
+
T\biggl({x-\mu\over\sigma};\alpha,p\biggr) G\biggl({x-\mu\over\sigma}\biggr)
\right\} \nonumber
\\[0,3cm]
& 
-
{\delta\sigma\over Z_{\boldsymbol{\theta}}}
\left\{
\mathbbm{E}\Big[\mathds{1}_{\{Y\leqslant  ({x-\mu\over\sigma\alpha})^2\}} G(-\alpha\sqrt{Y})\Big]
\mathds{1}_{\{x< \mu\}}
+
\mathbbm{E}\Big[\mathds{1}_{\{Y\leqslant  ({x-\mu\over\sigma\alpha})^2\}} G(\alpha\sqrt{Y})\Big]	
\mathds{1}_{\{x\geqslant \mu\}}
\right\},
\end{align}
for each $x\in\mathbbm{R}$. For more details, see Corollary \ref{prop-gen}. 

Taking $x=\mu$ in \eqref{CDF}, we get
$
F(\mu;\boldsymbol{\theta})
=
\sigma
\{
{\rho} G(0)
+
{\delta}
\mathbbm{E}[G(-\alpha\sqrt{Y})]
\}
/Z_{\boldsymbol{\theta}}.
$
Let $W$ be a random variable with corresponding CDF $G$.
If the distribution $G$ of $W$ is symmetric about zero, then 
$
Z_{\boldsymbol{\theta}}
=
\rho\sigma
+
2\delta\sigma 
{\mathbbm{E}}
[
G(-\alpha \sqrt{Y})
]
$,
$G(0)=1/2$, and then $F(\mu;\boldsymbol{\theta})=1/2$. So, in this case, $\mu$ is location parameter of $X$.
Moreover, by letting $x\to\infty$ in \eqref{CDF}, a simple observation shows that $F(x;\boldsymbol{\theta})$ tends to 1, showing that the parametric function in \eqref{pdfbgev2} is in fact a PDF.
\smallskip

Some natural examples of  kernel densities $g$ to be plugged into \eqref{pdfbgev2}, with $p$ given, where trimodality shape is observed, are presented in Table \ref{table:1}.
\begin{table}[H]
	\caption{Some kernel densities $(g)$ that generate multimodality in the model \eqref{pdfbgev2}.} \vspace*{0.15cm}
	\centering 
	\begin{tabular}{llllll} 
		\hline
		Distribution 
		& $g$ & $G$ & $\boldsymbol{\xi}$ & $D$ 
		\\ [0.5ex] 
		\noalign{\hrule height 0.7pt}
		Trimodal Gumbel 
		& ${\rm e}^{-x-{\rm e}^{-x}}$ & ${\rm e}^{-{\rm e}^{-x}}$  & $-$ & $\mathbbm{R}$ 
		\\ [1ex] 
		Trimodal Laplace
		& ${1\over 2}{\rm e}^{-\vert x\vert}$  & $1+{1\over 2}{\rm e}^{-x}[\mathds{1}_{\{x\leqslant 0\}}-\mathds{1}_{\{x\geqslant 0\}}]$  & $-$ & $\mathbbm{R}$ 
		\\ [1ex]   
		Trimodal Logistic
		& ${{\rm e}^{-x}\over (1+{\rm e}^{-x})^2}$  & ${1\over 1+{\rm e}^{-x}}$ &  $-$ & $\mathbbm{R}$  
		\\ [1ex] 
		Trimodal Cauchy
		& ${1\over \pi(1+x^2)}$  & ${1\over\pi}{\rm arctan}(x)+{1\over 2}$ & $-$ & $\mathbbm{R}$ 
		\\ [1ex] 
		Trimodal Student-$t$
		& ${\Gamma({\nu+1\over 2})\over\sqrt{\nu\pi} \Gamma({\nu\over 2}) }(1+{x^2\over\nu})^{-(\nu+1)/ 2}$  & ${1\over 2}+\Gamma({\nu+1\over 2}) {_2F_1({1\over 2},{\nu+1\over 2};{3\over 2}; -{x^2\over \nu} )\over\sqrt{\nu\pi} \Gamma({\nu\over 2})}$ & $\nu>0$ & $\mathbbm{R}$
		\\ [1ex]  
		Trimodal Normal
		& $\phi(x)={1\over\sqrt{2\pi}} {\rm e}^{-x^2/2}$  & $\Phi(x)=\int_{-\infty}^x \phi(t){\rm d}t$ & $-$ & $\mathbbm{R}$ 
		\\ [1ex] 
		\hline
	\end{tabular}
	\label{table:1} 
\end{table}
\noindent
Here $_2F_1$ is the hypergeometric function and ${\rm erf}(x)={2}\int_{0}^{x} {\rm e}^{-t^2} {\rm d}t/\sqrt{\pi}$ is the error function (also called the Gauss error function).

\section{Structural properties}
\label{Structural properties}
\noindent
In this section, some basic properties such as trimodality for symmetric kernels, moments and truncated moments, and entropies for $X\sim {\rm TD}(\boldsymbol{\theta})$ are discussed in detail.

\subsection{Trimodality for a class of symmetric kernel densities}
In this subsection, we suposse that the kernel density $g$ in \eqref{pdfbgev2} has the following form
\begin{align}\label{g-def}
	g(x)=g(0)\, {\rm e}^{^{\textstyle - \int_{0}^{x}t\mathfrak{h}(t^2)\, {\rm d}t}}, \quad x\in\mathbbm{R},
\end{align}
for some positive real function $\mathfrak{h}$ such that the integral $\int_{0}^{x}t\mathfrak{h}(t^2)\, {\rm d}t$ exists. 	
Notice that \eqref{g-def} is equivalent to
\begin{align}\label{der-condition}
g'(x)=-x\mathfrak{h}(x^2) g(x), \quad x\in\mathbbm{R}.
\end{align}
It is immediate to verify that $g$, as defined in \eqref{g-def},  is symmetric about zero, that is, $g(x)=g(-x)$ on the real line $D=\mathbbm{R}$. 
For example, in the Laplace, Cauchy, Student-$t$ and Normal kernel densities (see Table \ref{table:1}) we have  $\mathfrak{h}(y)=1/\sqrt{y}$,  $\mathfrak{h}(y)=2/(1+y)$, $\mathfrak{h}(y)=(\nu+1)/(\nu+{y})$ and $\mathfrak{h}(y)=1$, $\forall y>0$, respectively. 

Moreover, we assume that 
$\mathfrak{h}$ has the following form
\begin{align}\label{h-definition}
\mathfrak{h}(y)={C\over (\alpha^2 A+y)^{\beta-p}} \quad \text{for } \, C\geqslant 1, A\geqslant 0 \ \text{and} \ \beta>p.
\end{align}
That is, $\mathfrak{h}(y)$, $y>0$, decays polynomially.
The function  $\mathfrak{h}(y)=1$ corresponding to the Normal distribution is not of the form \eqref{h-definition}, then the next result cannot be applied and a separate study must be carried out.
In this paper, the Gaussian case 
will be studied in detail in Section \ref{Sect:4}.

For a formal proof of following lemma, see Section \ref{Ap-1} of Appendix.
\begin{lemma}\label{def-R-0}
	Let  $\mathfrak{h}$ be as in \eqref{h-definition}.
For some $\rho>0$ the function $\mathfrak{R}$, defined by
	\begin{align*}
	\mathfrak{R}(y)
	=
	{2\delta}\biggl[ 
	{(1/\alpha^2)^p\over\Gamma(p)}\, y^{p-1} {\rm e}^{-y/\alpha^2}\biggr]
	-
	\biggl[\rho+\delta\, {\gamma(p,y/\alpha^2)\over\Gamma(p)}\biggr]\mathfrak{h}(y), \quad y>0,
	\end{align*}
	has at most two real roots.
\end{lemma}

\begin{proposition}\label{prop-crit-0}
	Let  $g$  be a kernel density as in \eqref{g-def}, with $\mathfrak{h}$ as in \eqref{h-definition}.
	A point $x\in\mathbbm{R}$ is a critical point of density \eqref{pdfbgev2} if $x=\mu$ or 
	$
	\mathfrak{R}[(x-\mu)^2/\sigma^2]=0,
	$
	where $\mathfrak{R}$ is as in Lemma \ref{def-R-0}.
\end{proposition}
\begin{proof}
The proof is immediate since,
by using Equation \eqref{der-condition}, the first-order derivative of $f(x;\boldsymbol{\theta})$, with respect to $x$, is given by
	$
	f'(x;\boldsymbol{\theta})
	=
	[(x-\mu)/\sigma] 
	g((x-\mu)/\sigma) 
	\mathfrak{R}[(x-\mu)^2/\sigma^2]/(\sigma Z_{\boldsymbol{\theta}})	
	$.
\end{proof}

\begin{theorem}[Uni- bi- or trimodality]\label{Teo-trimodality}
	If $X\sim {\rm TD}(\boldsymbol{\theta})$
	then the following hold:
	\begin{itemize}
	\item[(1)] If $\mathfrak{R}$ has no real roots then  $f(x;\boldsymbol{\theta})$ is unimodal with mode $x=\mu$.
	\item[(2)] If $\mathfrak{R}$ has one real root then $f(x;\boldsymbol{\theta})$ is bimodal with minimum point $x=\mu$.
	\item[(3)] If $\mathfrak{R}$ has two distinct real roots then $f(x;\boldsymbol{\theta})$ is trimodal where $x=\mu$ is one of the modes.
\end{itemize}
\end{theorem}
\begin{proof}
It is clear that if $\mathfrak{R}$ has no real roots, by Proposition \ref{prop-crit-0}, $x=\mu$ is the only critical point of the density $f$. Since $\lim_{x\to\pm \infty}f(x;\boldsymbol{\theta})=0$, the point $x=\mu$ is a mode. This proves Item (1).

In order to prove Item (2), we suppose that $\mathfrak{R}$ has one real root, denoted by $a$. By Proposition \ref{prop-crit-0}, it follows that $x=\mu$ and $x=\mu\pm\sigma \sqrt{a}$ are three critical points of  $f$. Since $\lim_{x\to\pm \infty}f(x;\boldsymbol{\theta})=0$, the point $x=\mu$ is a minimum and $x=\mu\pm\sigma \sqrt{a}$ are two symmetrical modes. This proves the second item.

Now, we assume that $\mathfrak{R}$ has two distinct real roots, denoted by $a$ and $b$. Without loss of generality, we can assume that $a<b$. Again, by Proposition \ref{prop-crit-0} we have that $x=\mu$, $x=\mu\pm\sigma \sqrt{a}$ and $x=\mu\pm\sigma \sqrt{b}$ are five critical points of $f$. Since $\lim_{x\to\pm \infty}f(x;\boldsymbol{\theta})=0$ and $a<b$, the critical points $x=\mu$ and $x=\mu\pm\sigma \sqrt{b}$ are modes and $x=\mu\pm\sigma \sqrt{a}$ are minimum points. Hence, the proof of Item (3) follows. 
\end{proof}


\subsection{Moments}

\begin{theorem}\label{prop-exp}
Let $X\sim {\rm TD}(\boldsymbol{\theta})$ and  $L:\mathbbm{R}\to\mathbbm{R}$ be a Borel-measurable function. Then, the expectation of random variable $L(X)$ with $X\leqslant b$ and $b\in\mathbbm{R}$, is given by 
\begin{align*}
\mathbbm{E}\big[\mathds{1}_{\{X\leqslant b\}} L(X)\big]
&=
{\rho\sigma\over Z_{\boldsymbol{\theta}}}\, 
\mathbbm{E}\big[\mathds{1}_{\{W\leqslant{b-\mu\over\sigma}\}}L(W_{\mu,\sigma})\big]
\\[0,3cm]
&+
{\delta\sigma\over Z_{\boldsymbol{\theta}}}
\left\{ 
\mathbbm{E}\left[\mathds{1}_{\{Y\leqslant ({b-\mu\over\alpha\sigma})^2,\, W\leqslant {b-\mu\over\sigma}\}} L(W_{\mu,\sigma})\right]
+
\mathbbm{E}\left[\mathds{1}_{\{Y\geqslant ({b-\mu\over\alpha\sigma})^2,\,W\leqslant -\alpha\sqrt{Y}\}} L(W_{\mu,\sigma})\right]
\right\}\mathds{1}_{\{b< \mu\}}
\\[0,3cm]
&+
{\delta\sigma\over Z_{\boldsymbol{\theta}}}
\left\{ 
\mathbbm{E}\left[\mathds{1}_{\{W\leqslant-\alpha\sqrt{Y}\}} L(W_{\mu,\sigma})\right]
+
\mathbbm{E}\left[\mathds{1}_{\{Y\leqslant ({b-\mu\over\alpha\sigma})^2,\,\alpha\sqrt{Y}\leqslant W\leqslant {b-\mu\over\sigma}\}} L(W_{\mu,\sigma})\right]
\right\}\mathds{1}_{\{b\geqslant \mu\}},
\end{align*}
\noindent
where $W_{\mu,\sigma}=\sigma W+\mu$, $W$ is a  continuous random variable with CDF $G$ (that for brevity we write  $W\stackrel{d}{=}G$), $Y\sim{\rm Gamma}(p,1)$, and $W$ and $Y$ are independent.
\end{theorem}
\begin{proof} By using the definition of expectation and by taking the change of variables $w=(x-\mu)/\sigma$ and 
	${\rm d}x=\sigma {\rm d}w$, we have
	\begin{align}\label{exp-b}
	\mathbbm{E}\big[\mathds{1}_{\{X\leqslant b\}} L(X)\big]
	&=
	{\rho\over Z_{\boldsymbol{\theta}}}\, 
	\int_{\sigma D+\mu}
	\mathds{1}_{\{x\leqslant b\}}
	L(x)
	g\biggl({x-\mu\over\sigma}\biggr)
	\, {\rm d}x	
\nonumber
\\[0,3cm]
&	
	+
	{\delta\over Z_{\boldsymbol{\theta}}}\, 
	\int_{\sigma D+\mu}
	\mathds{1}_{\{x\leqslant b\}}
	L(x)
	T\biggl({x-\mu\over\sigma};\alpha,p\biggr)
	g\biggl({x-\mu\over\sigma}\biggr)
	\, {\rm d}x \nonumber
		\\[0,3cm]
	&=
	{\rho\sigma\over Z_{\boldsymbol{\theta}}}
\int_{D}
	\mathds{1}_{\{w\leqslant {b-\mu\over\sigma}\}}
L(\sigma w+\mu)
g(w)
 {\rm d}w	
+
{\delta\sigma\over Z_{\boldsymbol{\theta}}}\
\iint\limits_{\substack{ \scriptscriptstyle w\leqslant {(b-\mu)/\sigma} \\ \scriptscriptstyle 0<y\leqslant w^2/\alpha^2
}} \mathds{1}_D(\omega) \tau(w,y) \,{\rm d}y {\rm d}w,
	\end{align}
where, for notational simplicity, we denote 
\begin{align*}
\tau(w,y)
=
L(\sigma w+\mu) g(w) \biggl[{y^{p-1}{\rm e}^{-y} \over \Gamma(p)}\biggr].
\end{align*}
There are two cases to consider according to whether $\xi\coloneqq(b-\mu)/\sigma<0$ or $\xi\coloneqq(b-\mu)/\sigma\geqslant 0$; see Figure \ref{Fig-Parabola} (a) and (b).
\begin{figure}[H]
	\centering
	\includegraphics[height=5.cm,width=14.3cm]{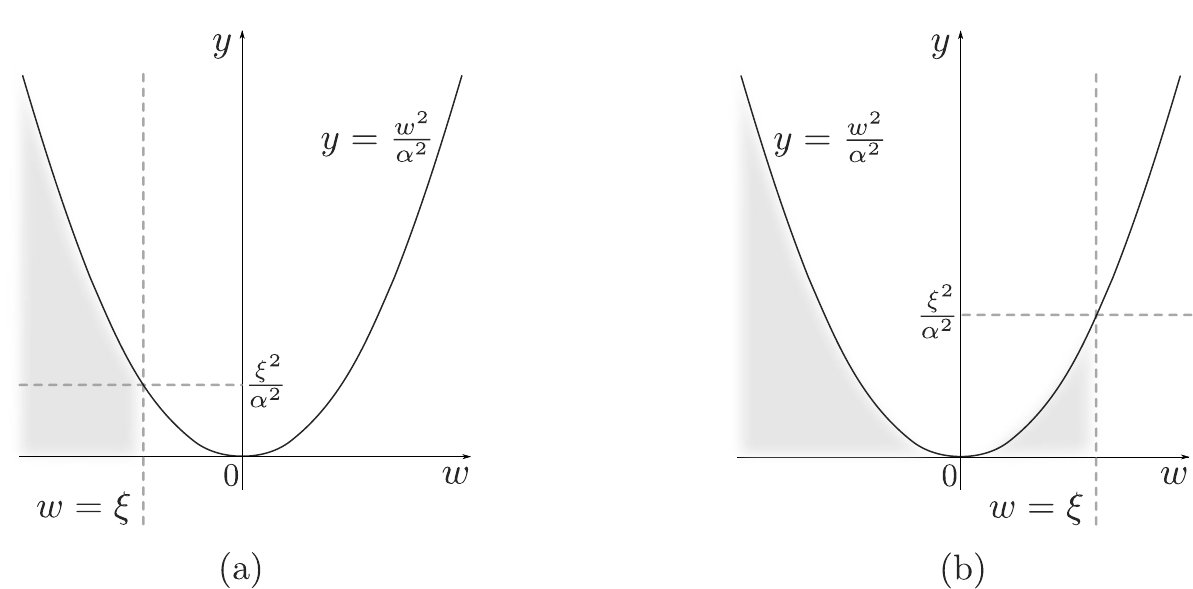}
	\caption{(a) $\{w\leqslant \xi, 0<y\leqslant w^2/\alpha^2, \xi<0\}$; \ (b) $\{w\leqslant \xi, 0<y\leqslant w^2/\alpha^2, \xi\geqslant 0\}$.}
	\label{Fig-Parabola}
\end{figure} 
\noindent
In the former case,
\begin{align*}
\iint\limits_{\substack{ \scriptscriptstyle w\leqslant {(b-\mu)/\sigma} \\ \scriptscriptstyle 0<y\leqslant w^2/\alpha^2
}} 
\mathds{1}_D(\omega) \tau(w,y) \,{\rm d}y {\rm d}w
&=
\int_{0}^{({b-\mu\over\alpha\sigma})^2}
\Biggl(
\int_{-\infty}^{{b-\mu\over\sigma}}
\mathds{1}_{D}(\omega)
\tau(w,y)
\, {\rm d}w
\Biggr)
{\rm d}y
\\[0,3cm]
&+
\int_{({b-\mu\over\alpha\sigma})^2}^{\infty}
\Biggl(
\int_{-\infty}^{-\alpha\sqrt{y}}
\mathds{1}_{D}(\omega)
\tau(w,y)
\, {\rm d}w
\Biggr)
{\rm d}y
\end{align*}
and in the latter case,	
\begin{align*}
\iint\limits_{\substack{ \scriptscriptstyle w\leqslant {(b-\mu)/\sigma} \\ \scriptscriptstyle 0<y\leqslant w^2/\alpha^2
}} \mathds{1}_D(\omega) \tau(w,y) \,{\rm d}y {\rm d}w
&=
\int_{0}^{\infty}
\Biggl(
\int_{-\infty}^{-\alpha\sqrt{y}}
\mathds{1}_{D}(\omega)
\tau(w,y)
\, {\rm d}w
\Biggr)
{\rm d}y
\\[0,3cm]
&+
\int_{0}^{({b-\mu\over\alpha\sigma})^2}
\Biggl(
\int_{\alpha\sqrt{y}}^{{b-\mu\over\sigma}}
\mathds{1}_{D}(\omega)
\tau(w,y)
\, {\rm d}w
\Biggr)
{\rm d}y.
\end{align*}

Hence, by combining the last two integral identities with \eqref{exp-b}, when $b<\mu$, we have
	\begin{align*}
\mathbbm{E}\big[\mathds{1}_{\{X\leqslant b\}} L(X)\big]
&=
{\rho\sigma\over Z_{\boldsymbol{\theta}}}
\int_{D}
\mathds{1}_{\{w\leqslant {b-\mu\over\sigma}\}}
L(\sigma w+\mu)
g(w)
{\rm d}w
+
{\delta\sigma\over Z_{\boldsymbol{\theta}}}
\int_{0}^{({b-\mu\over\alpha\sigma})^2}
\Biggl(
\int_{-\infty}^{{b-\mu\over\sigma}}
\mathds{1}_{D}(\omega)
\tau(w,y)
\, {\rm d}w
\Biggr)
\, {\rm d}y
\\[0,3cm]	
&+
{\delta\sigma\over Z_{\boldsymbol{\theta}}}
\int_{({b-\mu\over\alpha\sigma})^2}^{\infty}
\Biggl(
\int_{-\infty}^{-\alpha\sqrt{y}}
\mathds{1}_{D}(\omega)
\tau(w,y)
\, {\rm d}w
\Biggr)
\, {\rm d}y
\end{align*}
and, when $b\geqslant \mu$,
	\begin{align*}
\mathbbm{E}\big[\mathds{1}_{\{X\leqslant b\}} L(X)\big]
&=
{\rho\sigma\over Z_{\boldsymbol{\theta}}}
\int_{D}
\mathds{1}_{\{w\leqslant {b-\mu\over\sigma}\}}
L(\sigma w+\mu)
g(w)
{\rm d}w
+
{\delta\sigma\over Z_{\boldsymbol{\theta}}}
\int_{0}^{\infty}
\Biggl(
\int_{-\infty}^{-\alpha\sqrt{y}}
\mathds{1}_{D}(\omega)
\tau(w,y)
\, {\rm d}w
\Biggr)
\, {\rm d}y
\\[0,3cm]	
&+
{\delta\sigma\over Z_{\boldsymbol{\theta}}}
\int_{0}^{({b-\mu\over\alpha\sigma})^2}
\Biggl(
\int_{\alpha\sqrt{y}}^{{b-\mu\over\sigma}}
\mathds{1}_{D}(\omega)
\tau(w,y)
\, {\rm d}w
\Biggr)
\, {\rm d}y.
\end{align*}
Then there are $W\stackrel{d}{=}G$ and $Y\sim{\rm Gamma}(p,1)$ independent so that, for $b<\mu$, 
\begin{align*}
\mathbbm{E}\big[\mathds{1}_{\{X\leqslant b\}} L(X)\big]
&=
{\rho\sigma\over Z_{\boldsymbol{\theta}}}\, 
\mathbbm{E}\big[\mathds{1}_{\{W\leqslant{b-\mu\over\sigma}\}}L(W_{\mu,\sigma})\big]
\\[0,3cm]
&+
{\delta\sigma\over Z_{\boldsymbol{\theta}}}
\left\{ 
\mathbbm{E}\left[\mathds{1}_{\{Y\leqslant ({b-\mu\over\alpha\sigma})^2,\, W\leqslant {b-\mu\over\sigma}\}} L(W_{\mu,\sigma})\right]
+
\mathbbm{E}\left[\mathds{1}_{\{Y\geqslant ({b-\mu\over\alpha\sigma})^2,\,W\leqslant -\alpha\sqrt{Y}\}} L(W_{\mu,\sigma})\right]
\right\}
\end{align*}
and, for $b\geqslant\mu$, 
\begin{align*}
\mathbbm{E}\big[\mathds{1}_{\{X\leqslant b\}} L(X)\big]
&=
{\rho\sigma\over Z_{\boldsymbol{\theta}}}\, 
\mathbbm{E}\big[\mathds{1}_{\{W\leqslant{b-\mu\over\sigma}\}}L(W_{\mu,\sigma})\big]
\\[0,3cm]
&+
{\delta\sigma\over Z_{\boldsymbol{\theta}}}
\left\{ 
\mathbbm{E}\left[\mathds{1}_{\{W\leqslant-\alpha\sqrt{Y}\}} L(W_{\mu,\sigma})\right]
+
\mathbbm{E}\left[\mathds{1}_{\{Y\leqslant ({b-\mu\over\alpha\sigma})^2,\,\alpha\sqrt{Y}\leqslant W\leqslant {b-\mu\over\sigma}\}} L(W_{\mu,\sigma})\right]
\right\}.
\end{align*}
This completes the proof.
\end{proof}

By taking $b\to\infty$ in Theorem \ref{prop-exp}, with ${(X-\mu)/\sigma}$ instead of $X$, we get
\begin{corollary}\label{prop-exp-1}
Under the hypotheses of Theorem  \ref{prop-exp}, 
\begin{align*}
\mathbbm{E}\biggl[L\biggl({X-\mu\over\sigma}\biggr)\biggr]
=
{(\rho+\delta)\sigma\over Z_{\boldsymbol{\theta}}}\, \mathbbm{E}[L(W)]
+
{\delta\sigma\over Z_{\boldsymbol{\theta}}}
\left\{
\mathbbm{E}\left[\mathds{1}_{\{W\leqslant-\alpha\sqrt{Y}\}} L(W)\right]
-
\mathbbm{E}\left[\mathds{1}_{\{W\leqslant\alpha\sqrt{Y}\}} L(W)\right]
\right\}.
\end{align*}
\end{corollary}

\begin{corollary}\label{prop-gen}
Under the hypotheses of Theorem  \ref{prop-exp}, if 
\begin{itemize}
\item[(a)] $b\to\infty$ and  $L(x)=1$, $\forall x\in \sigma D+\mu$, then the formula \eqref{normalizator} for the normalizing factor $Z_{\boldsymbol{\theta}}$ is obtained.
\item[(b)] $b=x\in\mathbbm{R}$ fixed and $L(x)=1$, $\forall x\in \sigma D+\mu$, then the formula \eqref{CDF} of the CDF $F(x;\boldsymbol{\theta})$ is obtained.
\end{itemize}
\end{corollary}
 
By taking  $b\to\infty$  in Theorem \ref{prop-exp}, with $L(x)=x^n$, $\forall x\in \sigma D+\mu$ and $n\geqslant 1$ integer, by a binomial expansion, we have the next formula for the moments. 
\begin{corollary}\label{moments}
The $n$-th moment of $X\sim {\rm TD}(\boldsymbol{\theta})$ is  written as
\begin{align*}
\mathbbm{E}(X^n)
=
\sum_{k=0}^{n}\binom{n}{k} \mu^{n-k}\sigma^{k}
\left\{
{
{({{\rho}+\delta})\sigma\over Z_{\boldsymbol{\theta}}}\, \mathbbm{E}(W^k)
+
{\delta \sigma\over Z_{\boldsymbol{\theta}}}
\left[
\mathbbm{E}\left(\mathds{1}_{\{W\leqslant-\alpha\sqrt{Y}\}} W^k\right)
-
\mathbbm{E}\left(\mathds{1}_{\{W\leqslant\alpha\sqrt{Y}\}} W^k\right)
\right]
}
\right\}.
\end{align*}
The above formula informs that the moments of $X$ (whenever  they exist) depends on the existence of  moments of  $\mathds{1}_{\{W\leqslant \pm\alpha\sqrt{Y}\}} W$, with $W\stackrel{d}{=}G$ and $Y\sim{\rm Gamma}(p,1)$ independent.
\end{corollary}

\subsection{Entropies}
The Tsallis \cite{Tsallis1988} entropy associated with the random variable $X\sim {\rm TD}(\boldsymbol{\theta})$ is defined as
\begin{align*}
S_q(X)=
\begin{cases}\displaystyle
-\int_{\sigma D+\mu} f^q(x;\boldsymbol{\theta}) \log_q f(x;\boldsymbol{\theta}) \, {\rm d}x, & \text{ if} \ q\neq 1,
\\[0,6cm]\displaystyle
-\int_{\sigma D+\mu} f(x;\boldsymbol{\theta}) 
\log f(x;\boldsymbol{\theta}) \, {\rm d}x,& \text{ if} \  q=1,
\end{cases}
\end{align*}
where, for $x>0$,
\begin{align}\label{log-def}
\log_q(x) =
\begin{cases}\displaystyle 
{x^{1-q}-1\over 1-q}, & \text{ if} \ q\neq 1,
\\[0,2cm]
\log(x), & \text{ if} \ q= 1,
\end{cases}
\end{align}
represents a Box-Cox transformation in statistics (often called deformed logarithm \cite{Ferrari2010,Tsallis2009}). Since $\log_q(x)\to \log(x)$ when $q\to 1$, we have $S_q(X)\to S_1(X)$ when $q\to 1$. That is, when $q\to 1$ the usual definition of Shannon's entropy $S_1(X)$ \cite{Shannon1948}  is recovered.

\begin{proposition}
Under the hypotheses of Theorem  \ref{prop-exp}, 
\begin{align*}
\mathbbm{E}\big[f^{q-1}(X;\boldsymbol{\theta})\big]
&=
{({\rho}+\delta)\sigma\over Z_{\boldsymbol{\theta}}}\, 
\mathbbm{E}\big[f^{q-1}(\sigma W+\mu;\boldsymbol{\theta})\big]
\\[0,2cm]
&+
{{\delta}\sigma \over Z_{\boldsymbol{\theta}}}\,
\left\{
\mathbbm{E}\left[\mathds{1}_{\{W\leqslant-\alpha\sqrt{Y}\}} f^{q-1}(\sigma W+\mu;\boldsymbol{\theta})\right]
-
\mathbbm{E}\left[\mathds{1}_{\{W\leqslant\alpha\sqrt{Y}\}} f^{q-1}(\sigma W+\mu;\boldsymbol{\theta})\right]
\right\},
\end{align*}	
\noindent
where $W\stackrel{d}{=}G$, $Y\sim{\rm Gamma}(p,1)$, and $W$ and $Y$ are independent.	
Consequently, the Tsallis entropy of $X$ (whenever it exists) depends on the existence of  truncated moments of $f^{q-1}(W;\boldsymbol{\theta})$.
\end{proposition}
\begin{proof}
By taking $b\to\infty$ in Theorem \ref{prop-exp}, with $L(x)=f^{q-1}(x;\boldsymbol{\theta})$, $\forall x\in \sigma D+\mu$,  the proof follows.
\end{proof}

For a formal proof of next result, see Section \ref{Ap-1} of Appendix.
\begin{proposition}\label{existence-tsallis}
If $S_q(W)$ exists and $q>0,$ then $S_q(X)$, with $X\sim {\rm TD}(\boldsymbol{\theta})$, also exists.
\end{proposition}

\begin{remark}\label{rem-H1}
As consequence of Proposition \ref{existence-tsallis}, by letting $q\to 1$, the Shannon entropy $S_1(X)$  exists whenever $S_1(W)$ also exists.
\end{remark}

For a rigorous proof of next result, see Section \ref{Ap-1} of  Appendix.
\begin{proposition}\label{prop-Shannon entropy}
Under the hypotheses of Theorem  \ref{prop-exp},  the Shannon entropy of $X\sim {\rm TD}(\boldsymbol{\theta})$ is written as 
	\begin{align*}
	S_1(X)&=
	\log(Z_{\boldsymbol{\theta}})
	-
	{
			{({\rho}+\delta)\sigma \over Z_{\boldsymbol{\theta}}}\,
			\mathbbm{E}\big[\log(\rho+\delta T(W;\alpha,p))\big]		
}
	\\[0,2cm]
	&
	-
		{
		{{\delta}\sigma \over Z_{\boldsymbol{\theta}}}\,
		\left\{
		\mathbbm{E}\left[\mathds{1}_{\{W\leqslant-\alpha\sqrt{Y}\}} \log(\rho+\delta T(W;\alpha,p))\right]
		-
		\mathbbm{E}\left[\mathds{1}_{\{W\leqslant\alpha\sqrt{Y}\}} \log(\rho+\delta T(W;\alpha,p))\right]
		\right\}
	}
	\\[0,2cm]
&
+
	{
	{({\rho}+\delta)\sigma\over Z_{\boldsymbol{\theta}}}\, 
	S_1(W)
	+
	{{\delta}\sigma\over Z_{\boldsymbol{\theta}}}\,
	\left\{
	{\mathbbm{E}}
	\left[
	S_1\big(\mathds{1}_{ \{W\leqslant -\alpha\sqrt{Y}\} } W\big)
	\right]
	-
	{\mathbbm{E}}
	\left[
	S_1\big(\mathds{1}_{ \{W\leqslant \alpha\sqrt{Y}\} } W\big)
	\right]
	\right\}
},
	\end{align*}
where $T$ was defined in \eqref{tranf-G}.
\end{proposition}

\subsection{Stochastic representation}\label{stocrep}
Let $h(u)$, $0<u<1$, be a PDF with corresponding CDF $H$. 
Let $\mathcal{S}: D\to (0,1)$ be an injective and increasing transformation, where $D$ is a non-empty set of $\mathbbm{R}$. We consider the following  CDF:
\begin{align}\label{F-t-CDF}
F(z)=\int_{0}^{\mathcal{S}(z)} h(u) {\rm d}u
=H(\mathcal{S}(z)),\quad z\in D.
\end{align}
We also define by $f$ to the corresponding PDF of $F$. That is, $F'(z)=f(z)=h(\mathcal{S}(z))\mathcal{S}'(z)$ for almost all $z\in D$.

We define the PDF $h$ as follows
\begin{align}\label{def-h-pdf}
h(u)
=
{\sigma\over Z_{\boldsymbol{\theta}}} \, 
{\big[
\rho+\delta
T\big(G^{-1}(u);\alpha,p\big)
\big]}, \quad 0<u<1, \rho\geqslant 0, \delta\geqslant 0, \alpha>0, \sigma>0,
\end{align}
where $G$ and $G^{-1}$, respectively, are the CDF defined in Section \ref{A class of CDF} and its inverse function, and $T$ is as in \eqref{tranf-G}. When $\delta=0$,  $h$ reduces to the continuous uniform distribution on the interval $(0,1)$.  The CDF $H$ of  $h$ is given by
\begin{align*} 
H(u)
&=
{\rho\over Z_{\boldsymbol{\theta}}}\, u
+
{\delta\over Z_{\boldsymbol{\theta}}}
\left\{	\mathbbm{E}\big[G(-\alpha\sqrt{Y})\big]
+
\mathcal{S}\big(G^{-1}(u);\alpha,p\big) u
\right\} \nonumber
\\[0,3cm]
& 
-
{\delta\over Z_{\boldsymbol{\theta}}}
\left\{
\mathbbm{E}\Big[\mathbbm{1}_{\{Y\leqslant [{G^{-1}(u)\over\alpha}]^2\}} G(-\alpha\sqrt{Y})\Big]
\mathbbm{1}_{\{G^{-1}(u)< 0\}}
+
\mathbbm{E}\Big[\mathbbm{1}_{\{Y\leqslant [{G^{-1}(u)\over\alpha}]^2\}} G(\alpha\sqrt{Y})\Big]	
\mathbbm{1}_{\{G^{-1}(u)\geqslant 0\}}
\right\},
\end{align*}
where $Y\sim {\rm Gamma}(p,1)$.

If $\mathcal{S}: D\to (0,1)$ is defined by $\mathcal{S}(z)=G(z)$, $\forall z\in D$, by \eqref{F-t-CDF},  the family of trimodal distributions in \eqref{pdfbgev2} is obtained. That is, 
\begin{align}\label{cdf-F}
F(x;\boldsymbol{\theta})
=
F\biggl({x-\mu\over \sigma}\biggr)
=
H\biggl(G\Big({x-\mu\over \sigma}\Big)\biggr)
\end{align}
and  $f(x;\boldsymbol{\theta})={1\over \sigma}\, f((x-\mu)/ \sigma)$.

If $U$ is distributed according to \eqref{def-h-pdf} and $X\sim {\rm TD}(\boldsymbol{\theta})$, then, by \eqref{cdf-F},
the random variable $X$ admits the following stochastic representation:
\begin{align*}
X=\mu+\sigma G^{-1}(U).
\end{align*}

\section{The Gaussian case}\label{Sect:4}

In this section, 
the standard normal kernel density $g(x)=\phi(x)$, $x\in\mathbbm{R}$, is plugged into \eqref{pdfbgev2}.
Some structural properties as modality, moments, entropies and rate of the distribution are discussed.

The CDF corresponding to the normal kernel density $g$ is $G(x)=\Phi(x)$, with
\begin{align}\label{error-function}
\Phi(x)
=
\int_{-\infty}^x \phi(t)\,{\rm d}t
=
{1\over 2}
\biggl[1\pm{\rm erf}\biggl(\pm{x\over\sqrt{2}}\biggl)\biggr],
\end{align} 
where ${\rm erf}(x)$ is the error function. So, in this section, we consider the following PDF
\begin{eqnarray}\label{pdfbgev3}
f(x;\boldsymbol{\theta})
=
{1\over Z_{\boldsymbol{\theta}}}\, 
\biggl[\rho+\delta T\biggl({x-\mu\over\sigma};\alpha,p\biggr)\biggr] \,
\phi\biggl({x-\mu\over\sigma}\biggr),  \quad x\in \mathbbm{R},
\end{eqnarray}
where $T$ is as in \eqref{tranf-G} and $Z_{\boldsymbol{\theta}}$ is the normalizing factor \eqref{normalizator}, which is given by
\begin{align}\label{const-norm-Gaussian}
Z_{\boldsymbol{\theta}}
=
(\rho+ \delta)\sigma+
\delta\sigma 
\left\{
{\mathbbm{E}}
\big[
\Phi(-\alpha \sqrt{Y})
\big]
-
{\mathbbm{E}}
\big[
\Phi(\alpha \sqrt{Y})
\big]
\right\}.
\end{align}
We will denote $X\sim {\rm TD}_{\Phi}(\boldsymbol{\theta})$ for a random variable $X$ that follows \eqref{pdfbgev3}.
Plots of the TD$_\Phi$ density, where trimodality is observed, are given in Figures \ref{Fig-trimodality}
and \ref{Fig-trimodality2}.

\begin{figure}[ht]
	\begin{subfigure}{.5\textwidth}
		\centering
		\includegraphics[width=.9\linewidth]{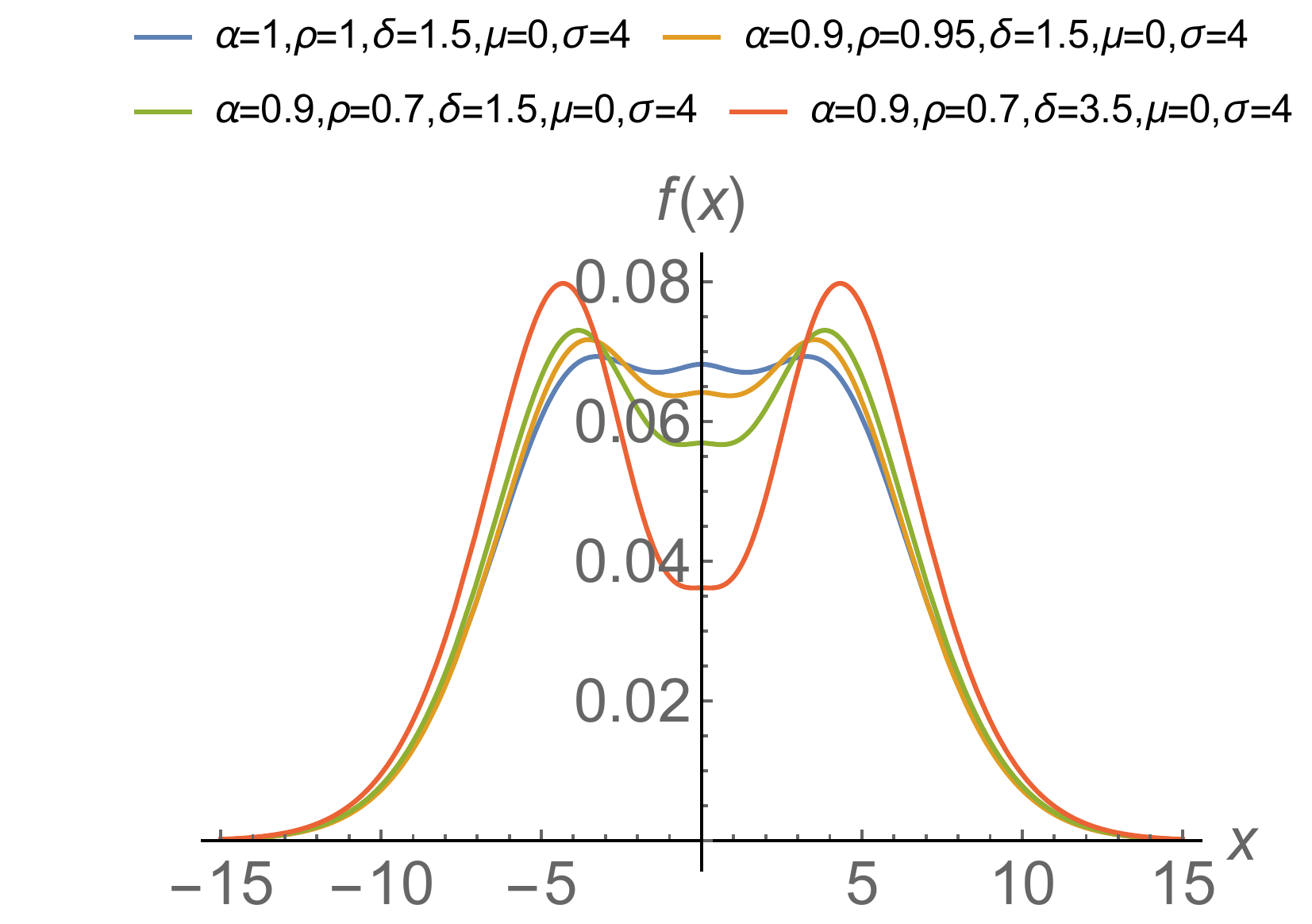}  
		\caption{The strict form of modality for trimodal normal distribution.}
		\label{Fig-trimodality}
	\end{subfigure}
	\begin{subfigure}{.5\textwidth}
		\centering
		\includegraphics[width=.9\linewidth]{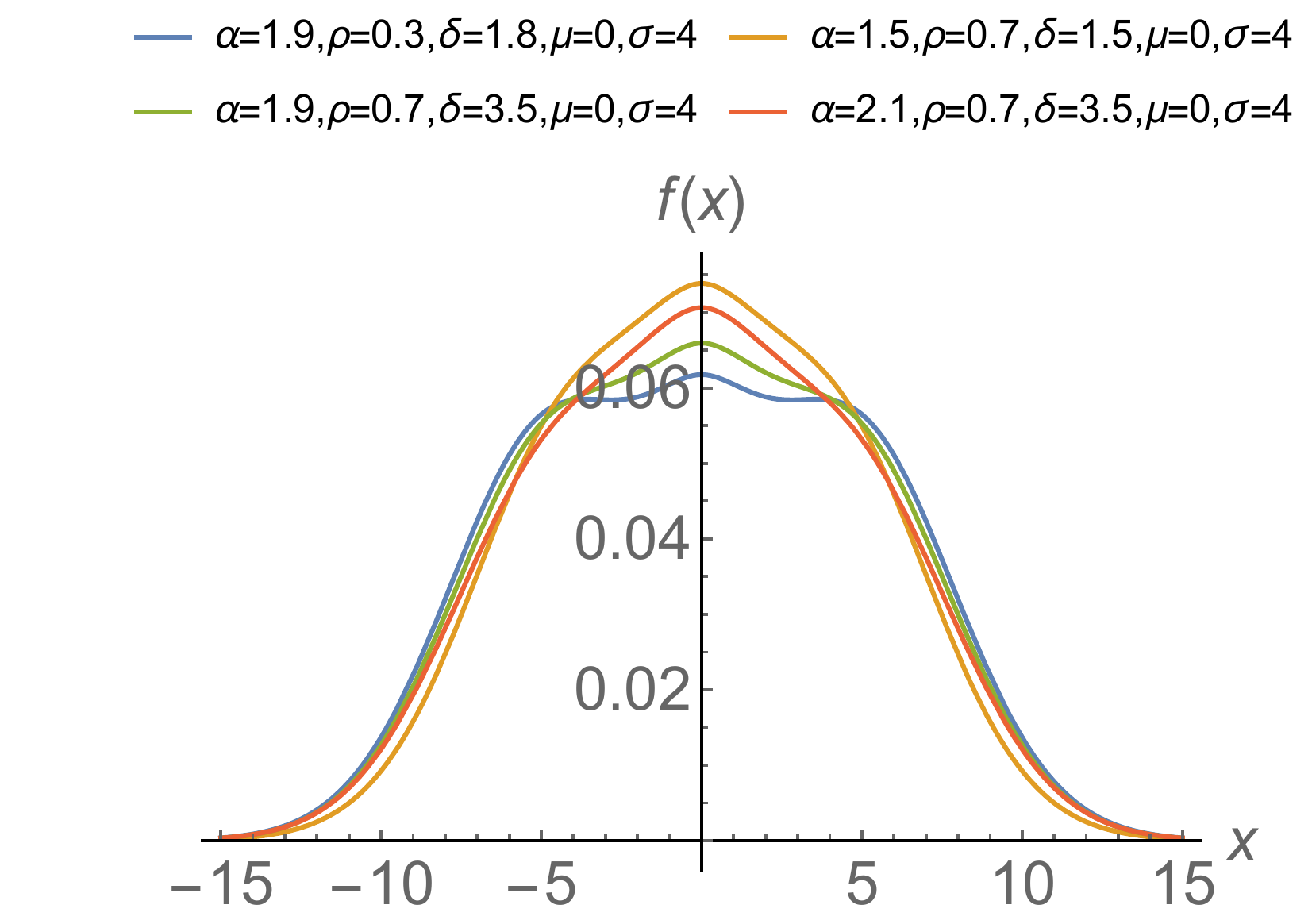}
		\caption{The soft form of modality for trimodal normal distribution.}
		\label{Fig-trimodality2}
	\end{subfigure}
	\caption{PDFs of  trimodal normal distribution.}
\end{figure}

\subsection{A study on the modality}

For a mathematical proof of next result, see Section \ref{Ap-1} of  Appendix.
\begin{lemma}\label{def-R}
	The function $R$, defined by
	\begin{align*}
	{R}(y)
	=
	{2\delta}\, 
	{{(1 / \alpha^2)^{p-1} {\rm e}^{-y/\alpha^2} \over\Gamma(p)}}
	-
	\biggl[\rho+\delta\, {\gamma(p,y/\alpha^2)\over\Gamma(p)}\biggr], \quad y>0,
	\end{align*}
	has at most two real roots.
\end{lemma}

\begin{remark}
Notice that when $\rho$ is sufficiently large, the function $R$ has no roots. For sufficiently small $\rho$, $R$ have one, or two roots depending on whether $p\leqslant 1$ or $p>1$. Furthermore, when $\delta=0$ (with $\rho>0$) or $\alpha\to\infty$, the function $R$ has no roots.
\end{remark}

\begin{proposition}\label{prop-crit}
	A point $x\in\mathbbm{R}$ is a critical point of density \eqref{pdfbgev3} if $x=\mu$ or 
	$
	R[(x-\mu)^2/\sigma^2]=0
	$,
	where $R$ is as in Lemma \ref{def-R}.
\end{proposition}
\begin{proof}
	The proof follows from identity 	$
	f'(x;\boldsymbol{\theta})
	=
	[(x-\mu)/\sigma] 
	g((x-\mu)/\sigma) 
	{R}[(x-\mu)^2/\sigma^2]/(\sigma Z_{\boldsymbol{\theta}})	
	$.
\end{proof}

\begin{theorem}[Uni- bi- or trimodality]
	Let $X\sim {\rm TD}_{\Phi}(\boldsymbol{\theta})$ and $R$ be the function defined in Lemma \ref{def-R}. The following hold:
	\begin{itemize}
		\item[(1)] If $R$ has no real roots then  $f(x;\boldsymbol{\theta})$ is unimodal with mode $x=\mu$.
		\item[(2)] If $R$ has one real root then $f(x;\boldsymbol{\theta})$ is bimodal with minimum point $x=\mu$.
		\item[(3)] If $R$ has two distinct real roots then $f(x;\boldsymbol{\theta})$ is trimodal where $x=\mu$ is one of the modes.
	\end{itemize}
\end{theorem}
\begin{proof} The proof follows the same steps of proof of Theorem \ref{Teo-trimodality} by taking $R$ instead $\mathfrak{R}$, and by using Proposition \ref{prop-crit} instead of Proposition \ref{prop-crit-0}.
\end{proof}

\subsection{The normalizing factor}
In what follows we find a closed expression for the normalizing factor $Z_{\boldsymbol{\theta}}$ in \eqref{pdfbgev3}. 
By \eqref{const-norm-Gaussian} it is necessary to calculate
${\mathbbm{E}}
[
\Phi(\pm\alpha \sqrt{Y})
]
$, where $Y\sim{\rm Gamma}(p,1)$.
Indeed, using the formula \eqref{error-function}
and then
taking the change of variables $z=\sqrt{y}$ and ${\rm d} y=2z {\rm d}z$, we obtain
\begin{align}
{\mathbbm{E}}
\big[
\Phi(\pm\alpha \sqrt{Y})
\big]
=
\int_{0}^{\infty} \Phi(\pm\alpha \sqrt{y}) \, 
{y^{p-1} {\rm e}^{-y}\over\Gamma(p)}
\, {\rm d}y 
&=
{1\over 2}
+
{1\over 2}\,
\int_{0}^{\infty} {\rm erf}\biggl(\pm{\alpha \sqrt{y}\over \sqrt{2}}\biggr) \, 
{y^{p-1} {\rm e}^{-y}\over\Gamma(p)}
\, {\rm d}y \label{exp-Phi-0}
\\[0,2cm]
&
=
{1\over 2}
+
\int_{0}^{\infty}   {\rm erf}\biggl(\pm{\alpha z\over \sqrt{2}}\biggr)\, 
{z^{2p-1} {\rm e}^{-z^2} \over \Gamma(p)} \, {\rm d}z. \label{exp-Phi}
\end{align}
By using the formula (see Item  8 in Subsection 4.3 of reference \cite{Ng-Geller1969}):
\begin{align*}
\int_{0}^{\infty}
{\rm erf}(ax) x^p {\rm e}^{-b^2 x^2} 
\, {\rm d}x
=
{a\over\sqrt{\pi}}\, b^{-p-2} \Gamma\Big({p\over 2}+1\Big)\,
_2F_1\biggl({1\over 2},{p\over 2}+1; {3\over 2};-{a^2\over b^2}\biggr),
\quad b^2>0, p>-2,
\end{align*}
where $_2F_1(a_1,a_2; b_1; x)$ is the generalized Hypergeometric function, we have
\begin{align*}
\int_{0}^{\infty}   {\rm erf}\biggl(\pm{\alpha z\over \sqrt{2}}\biggr)\, 
{z^{2p-1} {\rm e}^{-z^2} \over \Gamma(p)} \, {\rm d}z
=
\pm 
{\alpha \Gamma(p+{1\over 2})\over\sqrt{2\pi}\,\Gamma(p)}\,
_2F_1\biggl({1\over 2},p+{1\over 2}; {3\over 2};-{\alpha^2\over 2}\biggr).
\end{align*}
Therefore,
\begin{align}\label{CDF21}
{\mathbbm{E}}
\big[
\Phi(\pm\alpha \sqrt{Y})
\big]
=
{1\over 2}
\pm
{\alpha \Gamma(p+{1\over 2})\over\sqrt{2\pi}\,\Gamma(p)}\,
_2F_1\biggl({1\over 2},p+{1\over 2}; {3\over 2};-{\alpha^2\over 2}\biggr).
\end{align}

Replacing \eqref{CDF21} 
in \eqref{const-norm-Gaussian}, we obtain the following closed expression for the normalizing factor $Z_{\boldsymbol{\theta}}$:
\begin{align}\label{partition-function}
Z_{\boldsymbol{\theta}}
=
(\rho+ \delta)\sigma-
{{2}\delta\sigma \alpha \Gamma(p+{1\over 2})\over\sqrt{2\pi}\, \Gamma(p)}\,
_2F_1\biggl({1\over 2},p+{1\over 2}; {3\over 2};-{\alpha^2\over 2}\biggr).
\end{align}

The flexibility of the TD model is shown. Note that the TD model can be unimodal, bimodal or trimodal. Figures \ref{Fig-trimodality} and \ref{Fig-trimodality2} show how the TD density function is influenced by parameters $\alpha$, $\rho$ and $\delta$.

\subsection{Cumulative distribution function}\label{cumfunTD}
To determine the CDF of $X\sim {\rm TD}_\Phi(\boldsymbol{\theta})$ we use formula \eqref{CDF}. So, by \eqref{CDF}, it is essential to calculate the expectation
$
\mathbbm{E}[\mathds{1}_{\{Y\leqslant ({x-\mu\over\sigma\alpha})^2\}} \Phi(\pm\alpha\sqrt{Y})],$
$x\in\mathbbm{R},
$
where $Y\sim{\rm Gamma}(p,1)$. 
Indeed, similarly to the one done in \eqref{exp-Phi}, 
\begin{align}\label{CDF22}
{\mathbbm{E}}
\Big[\mathds{1}_{\{Y\leqslant ({x-\mu\over\alpha\sigma})^2\}}
\Phi(\pm\alpha \sqrt{Y})
\Big]
=
{1\over 2}\, T\biggl({x-\mu\over\sigma};\alpha,p\biggr)
+
{1 \over \Gamma(p)} \,
I\biggl({x-\mu\over\alpha\sigma};\pm{\alpha \over \sqrt{2}},2p-1,1\biggr),
\end{align}
where
$
	I(u;a,b,c)
	=
	\int_{0}^{u}
	{\rm erf}(ax) x^b {\rm e}^{-c^2 x^2} 
	\, {\rm d}x, \ u>0, a\in\mathbbm{R}, b>0, c>0.
$
%
Since $\vert{\rm erf}(x)\vert \leqslant 1$, we get $\vert I(u;a,b,c)\vert\leqslant I(u;0,b,c)<\infty$, and therefore,  $I(u;a,b,c)$ always exists. In general, closed form solutions for the definite integral $I(u;a,b,c)$ are not available in terms of commonly used functions. 

Replacing the formulas \eqref{CDF21} and \eqref{CDF22} in \eqref{CDF}, we obtain the following closed formula for the  CDF of $X\sim {\rm TD}_\Phi(\boldsymbol{\theta})$:
{\small 
\begin{align*}
F(x;\boldsymbol{\theta})
=
\begin{cases}
 \displaystyle 
{\sigma\over Z_{\boldsymbol{\theta}}}
\displaystyle \biggl[\rho+\delta T\biggl({x-\mu\over\sigma};\alpha,p\biggr)\biggr]
\Phi\biggl( {x-\mu\over\sigma}\biggr)
+
{\delta\sigma\over Z_{\boldsymbol{\theta}}}
\biggl[
{1\over 2}
-
{\alpha \Gamma(p+{1\over 2})\over\sqrt{2\pi}\,\Gamma(p)}\,
_2F_1\biggl({1\over 2},p+{1\over 2}; {3\over 2};-{\alpha^2\over 2}\biggr)
\biggr]&
\\[0,5cm]
\displaystyle 
\hspace{4.1cm}
-
{\delta\sigma\over Z_{\boldsymbol{\theta}}}\,
\left[
{1\over 2}\, T\biggl({x-\mu\over\sigma};\alpha,p\biggr) \nonumber
-
{1\over\Gamma(p)}\,
I\biggl({x-\mu\over\alpha\sigma};{\alpha \over \sqrt{2}},2p-1,1\biggr)
\right], & {\rm if}\, x< \mu,
\\[1cm]
\displaystyle 
{\sigma\over Z_{\boldsymbol{\theta}}}
\displaystyle \biggl[\rho+\delta T\biggl({x-\mu\over\sigma};\alpha,p\biggr)\biggr]
\Phi\biggl( {x-\mu\over\sigma}\biggr)
+
{\delta\sigma\over Z_{\boldsymbol{\theta}}}
\biggl[
{1\over 2}
-
{\alpha \Gamma(p+{1\over 2})\over\sqrt{2\pi}\,\Gamma(p)}\,
_2F_1\biggl({1\over 2},p+{1\over 2}; {3\over 2};-{\alpha^2\over 2}\biggr)
\biggr]  &
\\[0,5cm]
\displaystyle 
\hspace{4.1cm}
-
{\delta\sigma\over Z_{\boldsymbol{\theta}}}\,
\left[
{1\over 2}\, T\biggl({x-\mu\over\sigma};\alpha,p\biggr) 
\nonumber
+
{1\over\Gamma(p)}\,
I\biggl({x-\mu\over\alpha\sigma};{\alpha \over \sqrt{2}},2p-1,1\biggr)
\right], & {\rm if}\, x\geqslant \mu,
\end{cases}
\end{align*}
}
\noindent
where $Z_{\boldsymbol{\theta}}$ and $T$ are as in \eqref{partition-function} and \eqref{tranf-G}, respectively.

\begin{remark}
As expected, since $\Phi$ has a symmetric distribution around $0$,
$F(\mu;\boldsymbol{\theta})={1/ 2}$ and then, $Q_2=\mu$ is the median and the mean for $X\sim {\rm TD}_\Phi(\boldsymbol{\theta})$. 
\end{remark}

\begin{remark}
Some examples where $I({x-\mu\over\alpha\sigma};\pm{\alpha \over \sqrt{2}},2p-1,1)$ in \eqref{CDF22} admits a closed form are:
\begin{enumerate}
	\item 
	By taking $\alpha=\sqrt{2}$ and by using the following formula (see Item  4 in Subsection 1.5.3, pp. 31, of \cite{Prudnikov2002}):
	\begin{align*}
	\int_{0}^{u}
	{\rm erf}(ax) x^\lambda {\rm e}^{-a^2 x^2} 
	\, {\rm d}x
	=
	{2a\over\sqrt{\pi}(\lambda+2)}\, u^{\lambda+2}\,
	_2F_2\biggl({1},{\lambda\over 2}+1; {3\over 2},{\lambda\over 2}+2;{a^2 u^2}\biggr),
	\quad \lambda>-2,
	\end{align*}
	where $_2F_2(a_1,a_2; b_1,b_2; x)$ is the generalized Hypergeometric, 
	we have   
	\begin{align*}
	I\biggl({x-\mu\over\alpha\sigma};\pm{\alpha \over \sqrt{2}},2p-1,1\biggr)
	=\pm 
	{2\over\sqrt{\pi}(2p+1)}\, 
	\biggl({x-\mu\over\alpha\sigma}\biggr)^{2p+1}
	\,_2F_2\biggl({1},p+{1\over 2}; {3\over 2},p+{3\over 2};
	\Big({\mu-x\over\alpha\sigma}\Big)^2\biggr).
	\end{align*}
	\item 
	By taking $\alpha=\sqrt{2}$ and $p=3/2$,
	\begin{eqnarray*}
	I\biggl({x-\mu\over\alpha\sigma};\pm{\alpha \over \sqrt{2}},2p-1,1\biggr)
	=
	\mp
	{1\over 2}\, \biggl({x-\mu\over \sigma\sqrt{2}}\biggr) 
	{\rm erf}\biggl({x-\mu\over \sigma\sqrt{2}}\biggr)
	\exp\biggl[{-\biggl({x-\mu\over \sigma\sqrt{2}}\biggr)^2}\biggr]
			\\[0,3cm]
	\pm
	{1\over 4\sqrt{\pi}}\,
	\biggl\{1-\exp\biggl[{-2\biggl({x-\mu\over \sigma\sqrt{2}}\biggr)^2}\biggr]\biggr\}
	\pm
	{\sqrt{\pi}\over 8}\, \biggl[{\rm erf}\biggl({x-\mu\over \sigma\sqrt{2}}\biggr)\biggr]^2.
	\end{eqnarray*}
\end{enumerate}
\end{remark}

\subsection{Moments and Shannon entropy}\label{Moments and Shannon entropy}
A routine calculation shows that the moments of $X\sim {\rm TD}_\Phi(\boldsymbol{\theta})$ can be written as (for more details, see Section \ref{Moments in the Gaussian case} in the Appendix)
\begin{align*}
\mathbbm{E}(X^n)
&=
{({{\rho}+\delta})\sigma\over Z_{\boldsymbol{\theta}}}
\sum_{k=0}^{n} 
\binom{n}{k} \mu^{n-k}\sigma^{k}\,
{2^{-k/2} k!\over (k/2)!}\,\mathds{1}_{\{k \, {\rm even}\}}
\\[0,3cm]
&
+
{\delta\sigma \over  Z_{\boldsymbol{\theta}}  \sqrt{2\pi}\Gamma(p)}
\sum_{k=0}^{n} 
\binom{n}{k} \mu^{n-k}\sigma^{k}
{\big[1-(-1)^{k-1}\big]}
\sum_{m=0}^{\lfloor k/2\rfloor}
\sum_{j=0}^{\lfloor k_m/2\rfloor}
{c_{k,m}d_{k_m,j}\,\alpha^{k_m-2j}\over(1+{\alpha^2\over 2})^{{k_m-2j\over 2}+p}}\,
\Gamma\biggl({k_m-2j\over 2}+p\biggr),
\end{align*}
where $p\geqslant 1$ is an integer, $k_m=k-2m-1$,
$Z_{\boldsymbol{\theta}}$ is as in \eqref{partition-function}, $c_{n,m}$ and $d_{n_m,j}$ are as in \eqref{exp-tri-G}.
In particular,
$\mathbbm{E}(X)=\mu$ and ${\rm Var}(X)=\sigma^{2}$.
As a consequence, from Corollary \ref{prop-exp-1} the closed expressions for skewness and kurtosis of random variable $X\sim {\rm TD}_\Phi(\boldsymbol{\theta})$ are easily obtained.
\bigskip 

On the other hand, a formula for the Shannon entropy is given by (for more details, see Section \ref{Entropy in the Gaussian case} in the Appendix)
	\begin{align*}
	S_1(X)&=
	\log(Z_{\boldsymbol{\theta}})
	-
	{
		{2({\rho}+\delta)\sigma \over Z_{\boldsymbol{\theta}}}\,
		\sum_{k=0}^{\infty} {1\over (2k+1)} 
		\biggl[
		1
		+
		{1\over (\rho+1)^{2k+1}} 
		\sum_{i=p}^{\infty}
		\widetilde{c}_{i,k}\,
		{2^{-i} (2i)!\over i!}
		\biggr]
	}
	\\[0,3cm]
	&
	+
	{2{\delta}\sigma \over Z_{\boldsymbol{\theta}}}
	\sum_{k=0}^{\infty} {1\over (2k+1)} 
	\biggl\{
	2\mathbbm{E}\big[\Phi(\alpha\sqrt{Y})\big]
	-
	1
	+
	{1\over (\rho+1)^{2k+1}} 
	\sum_{i=p}^{\infty}
	\widetilde{c}_{i,k} 
	\mathbbm{E}\big(\mathds{1}_{\{\vert Z\vert\leqslant \alpha\sqrt{Y}\}} {Z^{2i}}\big)
	\biggr\}
	\\[0,3cm]
	&
	+
	{({\rho}+\delta)\sigma\over Z_{\boldsymbol{\theta}}}\, 
	\biggl[\log(\sqrt{2\pi})+{1\over 2}\biggr]
	+
	{{\delta}\sigma\over Z_{\boldsymbol{\theta}}}
	\left\{
	{\alpha\Gamma(p+{1\over 2})\over \sqrt{2\pi}\Gamma(p)}\,
	{1\over (1+{\alpha^2\over 2})^{p+{1\over 2}}} 
	-
	\biggl[\log(\sqrt{2\pi})+{1\over 2}\biggr]
	\big(
	2
	{\mathbbm{E}}
	\big[
	\Phi(\alpha \sqrt{Y})
	\big]
	-1
	\big)
	\right\},
	\end{align*}
\noindent
where $Z_{\boldsymbol{\theta}}$ is as in \eqref{partition-function}, the coefficients $\widetilde{c}_{i,k}$'s are determined by relation	\eqref{c-srec}, ${\mathbbm{E}}\big[\Phi(\alpha\sqrt{Y})\big]$
is given in \eqref{CDF21}, and
$
\mathbbm{E}(\mathds{1}_{\{\vert Z\vert\leqslant \alpha\sqrt{Y}\}} {Z^{2i}})
=
\mathbbm{E}(\mathds{1}_{\{Z\leqslant \alpha\sqrt{Y}\}} {Z^{2i}})
-
\mathbbm{E}(\mathds{1}_{\{Z\leqslant -\alpha\sqrt{Y}\}} {Z^{2i}})
$
is determined by \eqref{exp-z-truncated} in the Appendix.

\subsection{Rate of the random variable  \texorpdfstring{$X\sim {\rm TD}_{\Phi}(\boldsymbol{\theta})$}{k}}
 Following the reference \cite{Klugman1998}, the rate of a continuous random variable $X$ is given by
	\begin{align*}
	\tau_X=-\lim_{x\to\infty} {{\rm d} \log f_X(x)\over {\rm d}x},
	\end{align*}
	where $f_X(x)$ denotes its respective PDF.
	
	A simple computation shows that, the  rate of $X\sim {\rm TD}_{\Phi}(\boldsymbol{\theta})$ is
	\begin{align*}
	\tau_{{\rm TD}_{\Phi}(\boldsymbol{\theta})}
	=
	{1\over \sigma}
	\lim_{x\to\infty}\left[
	{x-\mu\over\sigma}
-
{2\over \alpha \Gamma(p)}\,
{({x-\mu\over\alpha\sigma})^{2p-1} {\rm e}^{-({x-\mu\over\alpha\sigma})^2}\over \rho+T({x-\mu\over\sigma};\alpha,p) }
	\right]
	=
\infty.
	\end{align*}
	Then, far enough out in the tail, the distribution of $X\sim{{\rm TD}_{\Phi}(\boldsymbol{\theta})}$ looks like a Normal distribution, as expected.
	In addition, we have some comparisons between the rate of $X\sim{{\rm TD}_{\Phi}(\boldsymbol{\theta})}$ with the rates of some random variables with known distributions in the literature: Inverse-gamma, Log-normal, Generalized-Pareto, {\rm BGumbel} \cite{OVBB21}, BWeibull \cite{VC20}, BGamma \cite{VFSPO20}, exponential 
	and Normal;
	\begin{align*}
	\tau_{{\rm InvGamma}(\alpha,\beta)}
	=
	\tau_{{\rm LogNorm}(\mu,\sigma^2)}
	&=
	\tau_{{\rm GenPareto}(\alpha,\beta, \xi)}
	= 
	\tau_{\rm BWeibull(\alpha<1,\beta,\delta)}=0
	\\
	&<
	\tau_{\rm BGumbel(\mu,\beta,\delta)}=
	\tau_{\rm BWeibull(\alpha=1,\beta,\delta)}=
	\tau_{{\rm BGamma}(\alpha,1/\beta,\delta)}=
	\tau_{{\rm exp}(1/\beta)}=1/\beta
	\\
	&<
	\tau_{\rm BWeibull(\alpha>1,\beta,\delta)}
	=
	\tau_{{\rm TD}_{\Phi}(\boldsymbol{\theta})}
	=
	\tau_{{\rm Normal}(\mu,\sigma^2)}=\infty.
	\end{align*}
	In other words, the tail of the normal distribution, of $X\sim {\rm TD}_{\Phi}(\boldsymbol{\theta})$ and of ${\rm BWeibull}(\alpha>1,\beta,\delta)$  are lighter than the tail of the other distributions specified above.

\section{The Model arising from a Generalized Mixture}
\label{Sect:5}
\noindent
The weighted distributions of random variable $(W-\mu)/\sigma$, with  weight function
$w_k$, have its PDF defined by
\begin{align*}
f_k(x;\mu,\sigma)
=
{1\over \int_D w_k(y) g(y) {\rm d}y}\,
w_k{\biggl({x-\mu\over\sigma}\biggr)}
g\biggl({x-\mu\over\sigma}\biggr),
\quad {x-\mu\over\sigma}\in D, \
k=0,1,\ldots.
\end{align*}

By using the power series expansion of the incomplete gamma function in
\eqref{series expansion},
note that the PDF $f(x;\boldsymbol{\theta})$ in \eqref{pdfbgev2}  interprets as an infinite (generalized) mixture of weighted distributions of $(W-\mu)/\sigma$ with the weight functions
$w_k(y)=y^{2k}$, $k=0,1,\ldots$, and with same parameter
vector $(\mu,\sigma)$. That is, 
\begin{align}\label{pdf-decomp}
f(x;\boldsymbol{\theta})
=
c_0
f_0(x;\mu,\sigma)
+
\sum_{k=p}^{\infty}
c_{k}
f_{k}(x;\mu,\sigma),
\end{align}
where the constants $c_0$, $c_k$, with $k=p,p+1,\ldots$, that depends only on $\boldsymbol{\theta}$, respectively, are given by
\begin{align*}
c_0=
{\rho\sigma\over Z_{\boldsymbol{\theta}}};
\quad 
c_k=
{\delta\sigma\over \Gamma(p)Z_{\boldsymbol{\theta}}}\,
{(-1)^{k-p} \mathbbm{E}(W^{2k})\over k(k-p)! \alpha^{2k}}, \quad k=p,p+1,\ldots.
\end{align*}
It is straightforward to verify that $c_0+\sum_{k=p}^{\infty}
c_{k}=1$. Note that the weights $c_k$ can take on negative values.

If $X\sim {\rm TD}_\Phi(\boldsymbol{\theta})$ and $L:\mathbbm{R}\to\mathbbm{R}$ is a Borel-measurable function, then, by representation \eqref{pdf-decomp}, the expectation of truncated random variable $\mathds{1}_{\{X\leqslant b\}} L(X)$, with $b\in\mathbbm{R}$, is given by
\begin{align*}
\mathbbm{E}\big[\mathds{1}_{\{X\leqslant b\}} L(X)\big]
=
c_0
\mathbbm{E}_0\big[\mathds{1}_{\{X\leqslant b\}} L(X)\big]
+
\sum_{k=p}^{\infty}
c_{k}
\mathbbm{E}_k\big[\mathds{1}_{\{X\leqslant b\}} L(X)\big].
\end{align*}
But this representation as infinite sum is not as informative as that representation of Theorem \ref{prop-exp}. For this reason in the paper, we had been concerned with providing mathematical representations that allow us to find the closed expressions for some characteristics of the distribution (such as normalizing factor, CDF, moments and entropies) as a function of the known mathematical functions.

\section{Inference}
\label{Sect:6}
\noindent
The maximum log$_q$-likelihood estimation (ML$_q$E) and Fisher information matrix for this estimation method are introduced.  ML$_q$E drops to the maximum likelihood estimation (MLE) if $q \to 1$ \cite{Cankaya2018,Ferrari2010,Fisher1925,Tsallis2009}.

\subsection{The Log$_q$-Likelihood Function}
For a random sample $X_1, \ldots , X_n$ from the random variable $X\sim {\rm TD}(\boldsymbol{\theta})$, with parameter vector $\boldsymbol{\theta}=(\mu,\sigma,\alpha,\rho, \delta)$, let us suppose that $x_1, \ldots, x_n$ are the observed values of $X_1, \ldots , X_n$.
The ML$_q$E of $\boldsymbol{\theta}$  is defined as
\begin{align}\label{derivative-logq}
l_q(\boldsymbol{\theta};\boldsymbol{x})
=
\sum_{i=1}^{n}
\log_q f(x_i;\boldsymbol{\theta}),
\quad \boldsymbol{x}=(x_1, \ldots, x_n)\in\mathbbm{R}^n,
\end{align}
where $\log_q(x)$ is the deformed logarithm defined in \eqref{log-def} \cite{Ferrari2010,Tsallis2009}.

By taking partial derivatives in \eqref{derivative-logq}, with respect to $\theta\in\{\mu,\sigma,\alpha,\rho, \delta\}$, we have
\begin{align*}
{\partial l_q(\boldsymbol{\theta};\boldsymbol{x})\over \partial\theta}
=
\sum_{i=1}^{n}
f^{1-q}(x_i;\boldsymbol{\theta})\,
{\partial \log f(x_i;\boldsymbol{\theta})\over \partial\theta},
\end{align*}
with
\begin{align}\label{log-ver}
\log f(x;\boldsymbol{\theta})
=
-\log(Z_{\boldsymbol{\theta}})
+ 
\log\biggl[\rho+\delta T\biggl({x-\mu\over\sigma};\alpha,p\biggr)\biggr] 
+
\log g\biggl({x-\mu\over\sigma}\biggr)
\end{align}
and $Z_{\boldsymbol{\theta}}$, $T$ are as in \eqref{normalizator} and \eqref{tranf-G}, respectively.
In general, the estimating equations for the parameters have the form
\begin{align}\label{MLqE}
\sum_{i=1}^{n}
f^{1-q}(x_i;\boldsymbol{\theta})\,
{\partial \log f(x_i;\boldsymbol{\theta})\over \partial\theta}=0,
\quad \theta\in\{\mu,\sigma,\alpha,\rho, \delta\},
\end{align}
where the first-order partial derivatives of the function $\log f(x;\boldsymbol{\theta})$ are given in Section  \ref{Elements of the score vector} of the Appendix.
A solution of the system of equations in \eqref{MLqE} is called ML$_q$E estimator.
It is not possible to derive analytical solution for the ML$_q$E $\widehat{\boldsymbol{\theta}}$.

\subsection{Fisher Information based on $\log_q$}
The definition of Fisher information (FI) based on $\log_q$ is given in \cite{Cankaya2018,Plastino1997}. The elements of FI matrix are defined by
\begin{align*}
[\mathcal{I}(\boldsymbol{\theta})]_{j,k}
=
\sum_{i=1}^{n}
\mathbbm{E}\bigg[
f^{1-q}(X_i;\boldsymbol{\theta})\,
{\partial \log f(X_i;\boldsymbol{\theta})\over \partial\theta_j}\,
{\partial \log f(X_i;\boldsymbol{\theta})\over \partial\theta_k}
\bigg],
\quad \theta_j, \theta_k\in\{\mu,\sigma,\alpha,\rho, \delta\},
\end{align*}
where 
$\log f(x;\boldsymbol{\theta})$ is given in \eqref{log-ver} and $f(x;\boldsymbol{\theta})$ is the parametric model in \eqref{pdfbgev2}. When the inverse of FI matrix exists, it is well-known that the diagonal elements of inverse of FI give ${\rm Var}(\widehat{\boldsymbol{\theta}})$ based on log$_q$ \cite{BercherFisq,Cankaya2018}. In general, there is no closed form expression for the FI matrix (see Appendices \ref{Elements of the score vector} and \ref{cisect}). It is clear that when $q=1$ in the above identity, under standard regularity conditions, we obtain the classical FI matrix  \cite{Fisher1925,Lehmann1998}.


\section{Application on real data sets}\label{appsect}
We present applications to illustrate the performance of the trimodal normal model compared with smooth kernel distribution as semiparametric distribution (\texttt{SmoothKernelDistribution}) and estimation of distribution which is performed by using a function named as \texttt{FindDistribution} embedded into  Mathematica 12.0 software to find an appropriate distribution for the data set. Since Mathematica software are capable for performing the optimization, bootstrap and also includes the numerical evaluation of hypergeometric function ($_2F_1$) while conducting the modelling the data sets, practitioners can use these codes for their aims in researches. The supplementary metarials  provide the building codes for  practitioners \cite{toolsinmathe,bootmathe,quantummathe}.

The Examples 1 and 2 represent the real data called as heterodatatrain\$V5 and heterodata\$V4, respectively in the  "Rmixmod" package at R software with version 4.1.3. The numbers of sample size $n$ are $300$ and $200$ for Examples 1 and 2, respectively.  The package "multilevel" includes data called as "bh1996". The columns 11 and 13 of bh1996 data represent the modality. The Examples 3 and 4  are for data sets with sample sizes $n=7382$. The $q$ values for Examples 1-4 are chosen as $0.98,0.95,0.98,0.99$, respectively, according to the GOFTs.

\begin{table}[H]
	\centering
	\caption{The models, the estimates of parameters, statistics and information criteria for  assesment of models}
		\resizebox{\linewidth}{!}{
	\begin{tabular}{c|cc|ccc|ccc}\hline
		\multicolumn{9}{c}{\bf{Example 1}}  \\ \hline
		& \multicolumn{2}{c}{Estimates}   & \multicolumn{3}{c}{Goodness of fit tests (GOFTs)}   & \multicolumn{3}{c}{log(L) \& Information criteria}    \\ \hline	
	PM& $\hat{\mu}$ & $\hat{\sigma}$   & KS &CVM  &AD  & log(L)& AIC & BIC  \\ \hline
		$^q$TD$_\Phi$ & -2.25117  & 1.32990  &  0.0778466 &0.551431  &620.933  & -570.632 & 1155.25 & 1172.67 \\
		SK	& -2.24486 & 1.62363  & {\bf 0.0680419}  &  {\bf0.402510} &  {\bf593.949} & -563.837 &  {\bf 1130.25}  & \bf{1137.49}\\
		EstD &  -2.22717	& 1.61703 &  0.0706750 & 0.410858 & 597.038 &  {\bf -563.777} &1134.51  & 1141.45\\
		TD$_\Phi$& -2.25803  & 1.38998 &   0.0764954  & 0.559258 & 615.978 & -570.669  &1147.46 & 1166.25\\
		\hline
				\multicolumn{9}{c}{\bf{Example 2}}  \\ \hline
		PM& $\hat{\mu}$ & $\hat{\sigma}$   & KS &CVM  &AD  & log(L)& AIC & BIC  \\ \hline
		$^q$TD$_\Phi$ & -1.94546 & 1.29762  &  0.120948 &0.988454  & 471.754 & -391.988 &794.657 & 814.271 \\
SK&-1.98323 &1.63448 & {\bf0.0862193} & {\bf0.420766}&{\bf400.882} & {\bf-365.779} &{\bf 740.391} & \bf{744.768}\\
		EstD & -1.96097	& 1.65964 &  0.0873551 & 0.43558 & 401.620 &   -370.828 &743.939  & 752.477\\
		TD$_\Phi$&-1.91852  & 1.47189 &   0.120597  & 0.927997 & 430.148 & -377.753  &772.283 & 787.288\\
		\hline
			\multicolumn{9}{c}{\bf{Example 3}}  \\ \hline
			PM& $\hat{\mu}$ & $\hat{\sigma}$   & KS &CVM  &AD  & log(L)& AIC & BIC  \\ \hline
		$^q$TD$_\Phi$ & 2.77324 & 0.856429  &  0.0350506 &1.66183  & 14510.4  & -9762.83 &19535.7 & 19570.2 \\
		SK	& 2.78055 &  0.915343 &  \bf{0.0150676}  &  \bf{0.293907} & 14661.0 &\bf{-9659.05}&  \bf{19322.1} & \bf{19335.9}\\
		EstD &  2.78047	& 0.909687 &  0.0326443 & 1.31802 & \bf{14397.0} &  -9764.47&19532.9  & 19546.8\\
		TD$_\Phi$& 2.77505  & 0.882003 &   0.0333040  & 1.52912 & {\it14414.5} &{\it-9761.59}  &{\it19533.2} & 19567.7\\
		\hline
		\multicolumn{9}{c}{\bf{Example 4}}  \\ \hline
		PM& $\hat{\mu}$ & $\hat{\sigma}$   & KS &CVM  &AD  & log(L)& AIC & BIC  \\ \hline
		$^{q}$TD$_\Phi$ & -0.001341870 & 0.865658  &  0.989796 &43.8267 & 677.270  & -9550.79 &19111.6& 19146.1 \\
		SK	& -0.000056456 &  0.889187 & \bf{0.981952}  &\bf{43.0978}  &\bf{600.388} &\bf{-9472.98} &  \bf{18950.0} & \bf{18963.8}\\
		EstD &  0.000581656	& 0.884030 &  0.988856 & 43.7301 & 663.308&  -9553.08&19110.2  & 19124.0\\
		TD$_\Phi$& -0.002632980  & 0.869382 &   0.989468  & 43.7947 & 672.560 &{\it-9550.61}  &{\it19111.2} & {\it19145.8}\\
		\hline
		\end{tabular}
	}\\
 \begin{tablenotes}\tiny
	\raggedright
	\item 	PM: (Semi)-parametric models
	\item $^q$TD$_\Phi$: Objective function $\log_q$ from trimodal normal distribution.
	\item SK: The smooth kernel distribution used Gaussian (normal) distribution (semiparametric model).
	\item EstD: The automatically chosen function by "FindDistribution" in Mathematica software.
	\item TD$_\Phi$: Objective function $\log$ from trimodal normal distribution.
	\item {\it Italic represents  closeness to the values produced by SK and EstD or almost best ones}
\end{tablenotes}
	\label{tableex1}
\end{table}
The location ($\mu$) and scale ($\sigma$) are important parameters to summarize the data set. The efficient estimations of these parameters depend on the chosen function used for modelling.  Table \ref{tableex1} provides them and other statistics for testing the modelling competence of the used functions. When all of PM ($^q$TD$_\Phi$, SK, EstD and TD$_\Phi$ ) in Table \ref{tableex1} are compared, SK and EstD are rival ones among models. On the other side, the statistics and information criteria of $^q$TD$_\Phi$ and TD$_\Phi$ can be near to SK and EstD as alternative models (see also discussion in Appendix \ref{bootsim}).

\begin{table}[H]
	\centering
	\caption{Standard errors of estimators and  estimates of shape and bimodality parameters in $^q$TD$_\Phi$ and TD$_\Phi$}
		\resizebox{\linewidth}{!}{
	\begin{tabular}{c|cc|ccc}\hline
		\multicolumn{6}{c}{\bf{Example 1}}  \\ \hline
			PM& $\sqrt{{\rm Var}(\hat{\mu})}$ & $\sqrt{{\rm Var}(\hat{\sigma})}$   & $\hat{\alpha}$  ($\sqrt{{\rm Var}(\hat{\alpha})}$) &$\hat{\rho}$ ($\sqrt{{\rm Var}(\hat{\rho})}$) &$\hat{\delta}$ ($\sqrt{{\rm Var}(\hat{\delta})}$)\\ \hline
	$^q$TD$_\Phi$  & 0.152203	   & 0.124745& 1.22607(0.253034)&1.06657(7.41619) &0.233824(1.63256) \\
	
	TD$_\Phi$	& 0.167649 & 0.143923  & 1.21113(0.257971) &  1.08073(7.41113) & 0.230591(1.58755) \\ \hline
			\multicolumn{6}{c}{\bf{Example 2}}  \\ \hline
			PM& $\sqrt{{\rm Var}(\hat{\mu})}$ & $\sqrt{{\rm Var}(\hat{\sigma})}$   & $\hat{\alpha}$  ($\sqrt{{\rm Var}(\hat{\alpha})}$) &$\hat{\rho}$ ($\sqrt{{\rm Var}(\hat{\rho})}$) &$\hat{\delta}$ ($\sqrt{{\rm Var}(\hat{\delta})}$)\\ \hline
		$^q$TD$_\Phi$  & 0.071482   & 0.103308 &  37.1884(205.348) &5.53004(3048.8)&50.6147(27963.2)  \\
		TD$_\Phi$	& 0.0985512  & 0.140166  & 31.9374(217.858)  &  11.6716(5270.27) & 48.0121(21714.5) \\ \hline
	\multicolumn{6}{c}{\bf{Example 3}}  \\ \hline
PM& $\sqrt{{\rm Var}(\hat{\mu})}$ & $\sqrt{{\rm Var}(\hat{\sigma})}$   & $\hat{\alpha}$  ($\sqrt{{\rm Var}(\hat{\alpha})}$) &$\hat{\rho}$ ($\sqrt{{\rm Var}(\hat{\rho})}$) &$\hat{\delta}$ ($\sqrt{{\rm Var}(\hat{\delta})}$)\\ \hline
$^q$TD$_\Phi$  &0.000889860   & 0.000992885 &  1.161980(0.00549161) &1.39357(0.0514929) &0.400877(0.0145018)\\

TD$_\Phi$	&0.000932255 & 0.000732616 & 0.888664(0.00562133)  &  1.17315(0.0706782) & 0.182165(0.0109832) \\ \hline
			\multicolumn{6}{c}{\bf{Example 4}}  \\ \hline
PM& $\sqrt{{\rm Var}(\hat{\mu})}$ & $\sqrt{{\rm Var}(\hat{\sigma})}$   & $\hat{\alpha}$  ($\sqrt{{\rm Var}(\hat{\alpha})}$) &$\hat{\rho}$ ($\sqrt{{\rm Var}(\hat{\rho})}$) &$\hat{\delta}$ ($\sqrt{{\rm Var}(\hat{\delta})}$)\\ \hline
$^q$TD$_\Phi$  &  0.0002215930  & 0.000690205 &  0.499788(0.00379787) &1.405910(0.00379787) &0.12482400(0.00240774)  \\

TD$_\Phi$	&0.0000789721  &0.000723072   & 0.509288(0.00275330)  &  0.346905(0.00799665) &0.0312284(0.000700832) \\ \hline
	\end{tabular}}\\
	\label{tableex2}
\end{table}
Table \ref{tableex2} shows the square root of variance of estimators and also the estimates of $\hat{\alpha}$, $\hat{\rho}$ and $\hat{\delta}$ for the examples analyzed.
\begin{table}[H]
	\centering
	\caption{The estimates for  $\hat{\mu}$ and $\hat{\sigma}$ from normal distribution and robust form}
	\begin{tabular}{c|c|cccc}
		&  & Example 1 & Example 2 & Example 3 & Example 4 \\ \hline
		\multirow{2}{*}{N} & $\hat{\mu}$  &-2.24483 & -1.96589 & 2.77989  & 1.24627 × 10$^{-12}$  \\
	   	& $\hat{\sigma}$ &  1.58461 & 1.58535 & 0.908344 &0.882701  \\ \hline
		\multirow{2}{*}{R} & Median & -2.34344 & -2.23693 & 2.83333 & 0.0309632 \\
		& MAD& 1.08095  &0.981103  &  0.611111& 0.607600 \\ \hline
	\end{tabular}
 \begin{tablenotes}\tiny
	\raggedright
	\item N: Normal distribution
	\item R: Robust statistics
	\item M: Median, MAD: Median absolute deviation (Median(|x-Median(x)|))
\end{tablenotes}
	\label{tableex3}
\end{table}

Table \ref{tableex3} introduces the basic statistics to see the role of  distribution with one mode property and trimodality. 

The estimates of $\hat{\mu}$ and $\hat{\sigma}$ from different PM show that we can have a clue to imply that the existence of modality can be observed, because the estimates of $\hat{\mu}$ and $\hat{\sigma}$ from SK as a smooth kernel technique based on working on the data-adaptive approach  (which is capable to fit the modality whether or not it exists in the reality--see also discussions on Appendix \ref{benefitof}) instead of parametric approach for modelling  can be close to the estimates of $\hat{\mu}$ and $\hat{\sigma}$ from $^q$TD$_\Phi$ and TD$_\Phi$. The kernel estimation method as a smoothing tehnique is  the best one generally. Even if 1000 replication for the different design of samplings constructed by use of bootstrap technique is applied, the numerical error(s) in computation for optimization can be tricker to consider and make an accurate judgement among the modelling performance of the used four models. For example, CDF and PDF of TD$_\Phi$ depend on the \texttt{Hypergeometric2F1} in Mathematica. Soft forms of PDF of TD$_\Phi$ in Figure \ref{Fig-trimodality2} and smooth kernel technique in Ref. \cite{smoothker} can be alternative to each other when the estimates of $\hat{\mu}$ and $\hat{\sigma}$, the statistics from GOFTs, the values of log(L) and IC are taken into account. On the other side, it is very difficult to know which function will be the best one for modelling when the data sets in the finite sample size are tried to be fitted by the functions. Even if the sample sizes of Examples 3 and 4 are 7382, eventually we have finite sample size whatever it is. The population in reality will not known exactly. Consequently, an alternative function can be necessary for driving modality via parameters $\rho$ and $\delta$.  This is the reason why we make a comparison between the SK and other parametric models to observe what and how the estimates of  $\hat{\mu}$ and $\hat{\sigma}$ will be changed if PDF are changed \cite{hub64,CanEEq}. Note that the modelling and numerical error(s) are topics which can affect each other.

\section{Conclusions}\label{concsect}
\noindent
Since recent times show that an increasing popularity has been observed in the modelling for data sets having modality, producing the trimodal form of any PDF has been proposed. The trimodal form is constructed by using the technique which includes the cumulative function of Maxwell distribution, the existing unimodal distribution and the corresponding normalizing constant of the proposed distribution. The properties of trimodality have been examined. The application of producing the trimodality has been conducted for the normal distribution which is symmetric and unimodal form with two parameters which are location and scale.  The properties of trimodal normal distribution have been examined. Thus, the applicability of this distribution have been tested. The trimodal normal distribution can have different forms such as strict and soft modalities to perform a precise fitting when there exists three modes in the empirical distribution of the data sets. 

A comparison among trimodal normal, the kernel type estimation method, the probable parametric distribution driven by Mathematica software  has been performed in order to make applications for numerical evaluation of  TD$_\Phi$. The $\log_q$-likelihood estimation method and its special form with $q \to 1$ have been used to estimate of parameters of trimodal normal distribution. The proofs, properties of TD$_\Phi$ distribution and codes used for application have been given by appendices if the researchers perform to model the data sets by use of TD$_\Phi$ distribution. 

The future will be an application on the different areas of statistics such as regression modelling, the tools in the multivariate statistics and other tools based on the distribution theory. The order statistic form of TD$_\Phi$ in the least informative distribution  will be  studied for the trimodal forms of the existing distributions in the applied field of science. Additionally, the precise modelling  for inliers into data sets can also be performed by use of trimodality, the generalized logarithms,  entropy functions,  order statistic and different estimation methods all together \cite{Cankayaorder,Can21frac}. A package in R software will be prepared for practitioners after the special function in R software are improved.

\paragraph{Acknowledgements}
We acknowledge the anonymous referees for their helpful
comments, suggestions and references provided in their reports.
R. V. thanks A. V. Medino, J. Roldan and  E. M. M. Ortega for partial discussions of Theorem \ref{prop-exp} and for general paper questions.

\paragraph{Disclosure statement}
There are no conflicts of interest to disclose.

\paragraph{Funding}
This study was financed in part by the Coordenação de Aperfeiçoamento de Pessoal de Nível Superior - Brasil (CAPES) (Finance Code 001).

\paragraph{ORCID}
Roberto Vila  \
\url{https://orcid.org/0000-0003-1073-0114}
\\
Victor Serra  \
\url{https://orcid.org/0000-0003-3061-2094}
\\
Mehmet N. Çankaya  \
\url{https://orcid.org/0000-0002-2933-857X}
\\
Felipe Quintino \ \url{https://orcid.org/
0000-0003-0286-0541}


\break 
\begin{appendices}

\section{Proof of some results of Sections \ref{Structural properties} and \ref{Sect:4}}\label{Ap-1}

\begin{proof}[Proof of Lemma \ref{def-R-0}]	
	Let $Y$ be a random variable with ${\rm Gamma}(p,1/\alpha^2)$ distribution. The corresponding PDF and CDF of  $Y$, respectively, are given by $f(y;\alpha,p)={(1/\alpha^2)^p}\, y^{p-1} {\rm e}^{-y/\alpha^2}/\Gamma(p)$ and
	$F(y;\alpha,p)={\gamma(p,y/\alpha^2)/\Gamma(p)}$. Therefore,
	\begin{align}\label{R-fucntion}
	\mathfrak{R}(y)
	=
	{2\delta} 
	f(y;\alpha,p)
	-
	[\rho+\delta F(y;\alpha,p)]\mathfrak{h}(y), \quad y>0.
	\end{align}
	
	Suppose $y\geqslant \alpha^2(1-A)$. Then 
	\begin{align*}
	{\rm e}^{(\alpha^2A+y)/\alpha^2}
	=
	\sum_{k=0}^{\infty} 
	{(\alpha^2A+y)^k/\alpha^{2k}\over k!}
	\geqslant 
	{(\alpha^2A+y)^{n+1}/\alpha^{2(n+1)} \over (n+1)!},
	\end{align*}
	for each $n\in\mathbbm{N}$, such that
	\begin{align*}
	{\rm e}^{-y/\alpha^2}
	\leqslant 
	{(n+1)! \alpha^{2(n+1)} {\rm e}^{A}
		\over (\alpha^2A+y)^{n+1} }.
	\end{align*}
	Hence,
	\begin{align}
	{2\delta} 
	f(y;\alpha,p)
	&\leqslant
	{2\delta}\biggl[ 
	{(1/\alpha^2)^p\over\Gamma(p)}\, y^{p-1}\, {(n+1)! \alpha^{2(n+1)} {\rm e}^{A}
		\over (\alpha^2A+y)^{n+1} }\biggr] 
	\nonumber
	\\[0,3cm]
	&\leqslant 
	{2\delta}\biggl[ 
	{(1/\alpha^2)^p\over\Gamma(p)}\, {(n+1)! \alpha^{2(n+1)} {\rm e}^{A}
		\over (\alpha^2A+y)^{n+1-p} }\biggr] 
	\leqslant
	{\rho C\over (\alpha^2A+y)^{\beta-p}}
	=
	\rho \mathfrak{h}(y), \label{ineq-f-h}
	\end{align}
	for $\beta<n+1$, $C\geqslant 1$ and some $\rho>0$, so that
	\begin{align}\label{rho-condition}
	\rho
	\geqslant 
	{2\delta}
	\biggl[ 
	{(1/\alpha^2)^p\over\Gamma(p)}\, 
	{(n+1)! \alpha^{2(n+1)} {\rm e}^{A}}
	\biggr].
	\end{align}
	Since $F(y;\alpha,p)$ is a CDF, from \eqref{ineq-f-h} it follows that
	\begin{align}\label{ineq-f_h}
	{2\delta} 
	f(y;\alpha,p)
	\leqslant
	[\rho+\delta F(y;\alpha,p)]\mathfrak{h}(y)
	\leqslant
	(\rho+\delta)\mathfrak{h}(y)
	, \quad \forall y\geqslant \alpha^2(1-A).
	\end{align}
	In other words, for every $y\geqslant \alpha^2(1-A)$ and $\rho$ large enough, the tail of Gamma distribution $2\delta f(y;\alpha,p)$ is lighter than the tail of $[\rho+\delta F(y;\alpha,p)]\mathfrak{h}(y)$ which decays polynomially.
	
	\smallskip
	On the other hand, it is well-known that:
	when $p< 1$, the Gamma distribution is exponentially shaped and asymptotic to both the vertical and horizontal axes;
	the Gamma distribution with shape parameter $p= 1$ and scale parameter $1/\alpha^2$ is the same as an exponential distribution of scale parameter (or mean) $1$;
	when $p$ is greater than one, the Gamma distribution assumes a maximum value (unimodal), but skewed shape. The skewness reduces as the value of $p$ increases.
	
	Based on the shapes of the Gamma distribution and on the inequality \eqref{ineq-f_h}, in Figure \ref{Fig-trimodality-2}, we graphically sketch the functions $2\delta f(y;\alpha,p)$ and $[\rho+\delta F(y;\alpha,p)]\mathfrak{h}(y)$, and consider all possible cases, by varying the parameters  $\alpha$, $\rho$ as in \eqref{rho-condition}, $\delta$ and $p$ (known), in which the graphs of $2\delta f(y;\alpha,p)$ and $[\rho+\delta F(y;\alpha,p)]\mathfrak{h}(y)$ intersect (or not). 
	\begin{figure}[H]
		\centering
		\includegraphics[scale=1.15]{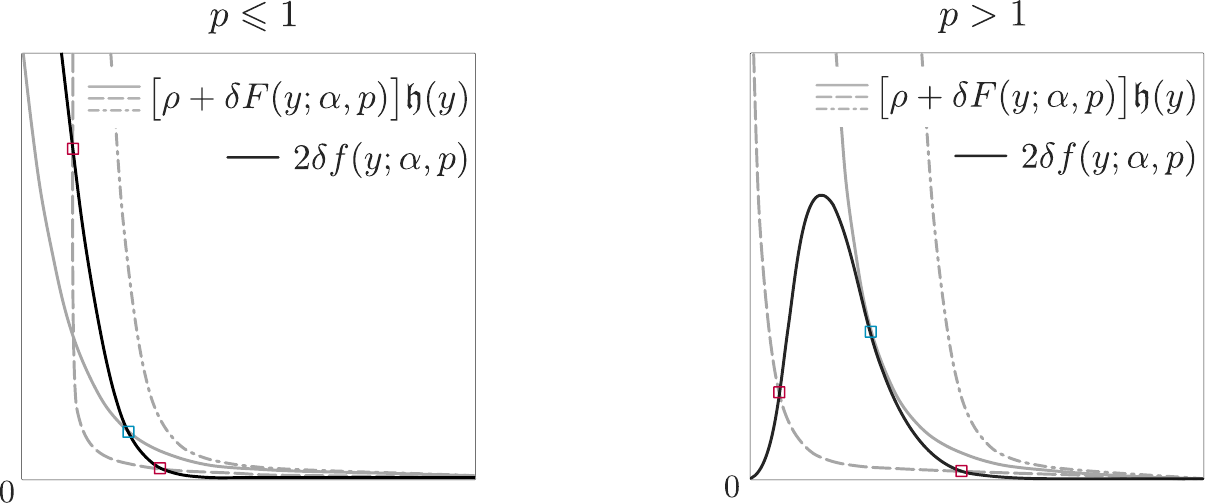}
		\caption{The graphs of $2\delta f(y;\alpha,p)$ and $[\rho+\delta F(y;\alpha,p)]\mathfrak{h}(y)$ have at most two common points.}
		\label{Fig-trimodality-2}
	\end{figure}
	\noindent
	Then, in both cases $p\leqslant 1$ or  $p> 1$,
	the function $\mathfrak{R}$ in \eqref{R-fucntion} 
	has at most two real zeros. 
	This completes the proof.
\end{proof}

\begin{proof}[Proof of Proposition \ref{existence-tsallis}]  A simple algebric manipulation shows that
	\begin{align*}
	S_q(X)=\dfrac{1}{q-1}\, \bigg[1-\int_{\sigma D+\mu} f^q(x;\boldsymbol{\theta}) \, {\rm d}x\bigg], \quad q\neq 1.
	\end{align*}	
	
	Since $0<T(x;\alpha,p)<1$ for almost all $x\in D$, we have 
	\begin{align*}
	0\leqslant f(x;\boldsymbol{\theta})\leqslant {(\rho+\delta)\over Z_{\boldsymbol{\theta}}}\, g\biggl({x-\mu\over\sigma}\biggr).
	\end{align*}
	Then, by using that the function $x\mapsto x^q$ is increasing, for $x>0$ and $q>0$, we have 
	\begin{align*}
	0\leqslant f^q(x;\boldsymbol{\theta})\leqslant {(\rho+\delta)^q\over Z^q_{\boldsymbol{\theta}}}\, g^q\biggl({x-\mu\over\sigma}\biggr),
	\end{align*}
	for all $x\in \sigma D+\mu$. Consequently,
	\begin{align*}
	\int_{\sigma D+\mu} f^q(x;\boldsymbol{\theta})\, {\rm d}x
	\leqslant
	{(\rho+\delta)^q\over Z^q_{\boldsymbol{\theta}}}\, 
	\int_{D}
	g^q\biggl({x-\mu\over\sigma}\biggr)
	\, {\rm d}x
	=
	{(\rho+\delta)^q\sigma\over Z^q_{\boldsymbol{\theta}}}\, \big[1-(q-1)S_q(W)\big],
	\quad q\neq 1.
	\end{align*}
	Hence, if $S_q(W)$ exists then of $S_q(X)$ also exists.
\end{proof}

\begin{proof}[Proof of Proposition \ref{prop-Shannon entropy}]
	By taking the logarithm of each side of	\eqref{pdfbgev2}, by definition of Shannon entropy, we have
	\begin{align}\label{eq-1}
	S_1(X)
	=
	\log(Z_{\boldsymbol{\theta}})
	+
	\mathbbm{E}\biggl[-\log\biggr(\rho+\delta T\biggl({X-\mu\over\sigma};\alpha,p\biggr)\biggr)\biggr]
	+
	\mathbbm{E}\biggl[-\log g\biggl({X-\mu\over\sigma}\biggr)\biggr].
	\end{align}
	
	By taking $L(x)=\log(\rho+\delta T(x;\alpha,p))$ and $L(x)=-\log g(x)$, $\forall x\in \sigma D+\mu$, in Corollary \ref{prop-exp-1}, respectively, we have
	\begin{align}\label{eq-2}
	\!\!\!\!\!\!\!
	\mathbbm{E}\biggl[-\log\biggr(\rho&+\delta T\biggl({X-\mu\over\sigma};\alpha,p\biggr)\biggr)\biggr]
	=
	-{
		{({\rho}+\delta)\sigma \over Z_{\boldsymbol{\theta}}}\,
		\mathbbm{E}\big[\log(\rho+\delta T(W;\alpha,p))\big]
	}	 \nonumber
	\\[0,2cm]
	&-
	{
		{{\delta}\sigma \over Z_{\boldsymbol{\theta}}}\,
		\left\{
		\mathbbm{E}\left[\mathds{1}_{\{W\leqslant-\alpha\sqrt{Y}\}} \log(\rho+\delta T(W;\alpha,p))\right]
		-
		\mathbbm{E}\left[\mathds{1}_{\{W\leqslant\alpha\sqrt{Y}\}} \log(\rho+\delta T(W;\alpha,p))\right]
		\right\}
	}
	\end{align}
	and
	\begin{align}\label{eq-3}
	\!\!\!\!
	\mathbbm{E}\biggl[-\log g\biggl({X-\mu\over\sigma}\biggr)\biggr]
	=
	{
		{({\rho}+\delta)\sigma \over Z_{\boldsymbol{\theta}}}\,
		S_1(W)
		+
		{{\delta}\sigma\over Z_{\boldsymbol{\theta}}}\,
		\left\{
		{\mathbbm{E}}
		\left[
		S_1\big(\mathds{1}_{ \{W\leqslant -\alpha\sqrt{Y}\} } W\big)
		\right]
		-
		{\mathbbm{E}}
		\left[
		S_1\big(\mathds{1}_{ \{W\leqslant \alpha\sqrt{Y}\} } W\big)
		\right]
		\right\}
	}.
	\end{align}
	By replacing the identities \eqref{eq-2} and \eqref{eq-3} in \eqref{eq-1}, the proof follows.
\end{proof}

\begin{proof}[Proof of Lemma \ref{def-R}]
	The PDF and CDF of $Y\sim{\rm Gamma}(p,1/\alpha^2)$, respectively, are given by $f(y;\alpha,p)={{(1/\alpha^2)^{p-1} {\rm e}^{-y/\alpha^2} /\Gamma(p)}}$ {and}  
	$F(y;\alpha,p)={\gamma(p,y/\alpha^2)/\Gamma(p)}$. Then, $R$ takes on the following form
	\begin{align}\label{def-R-f}
	{R}(y)
	=
	{2\delta} 
	f(y;\alpha,p)
	-
	[\rho+\delta F(y;\alpha,p)].
	\end{align}
	\begin{figure}[H]
		\centering
		\includegraphics[scale=1.15]{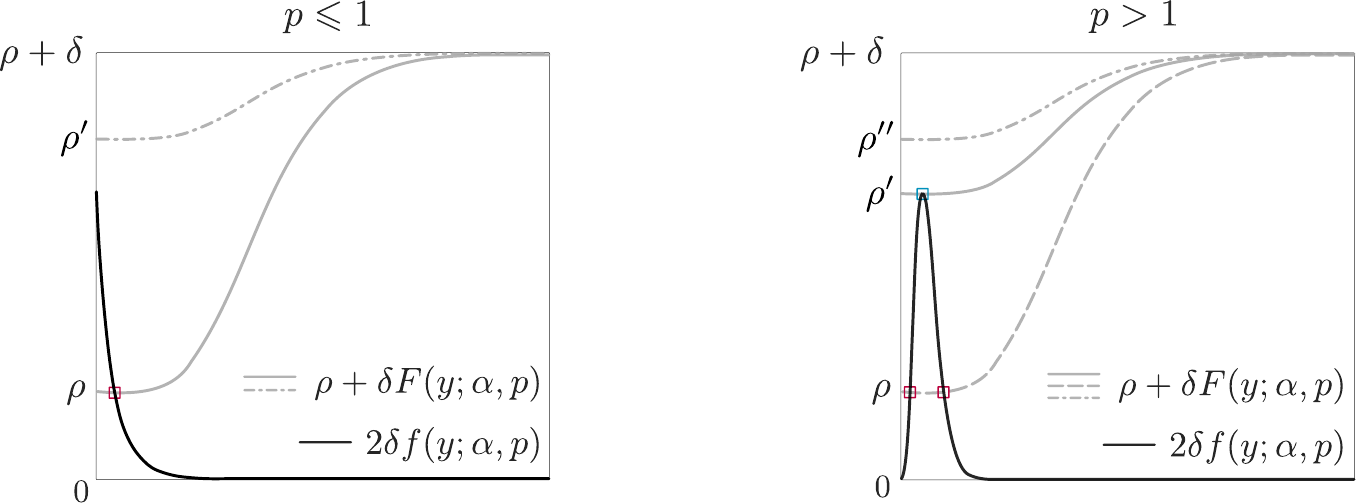}
		\caption{
			The graphs of $2\delta f(y;\alpha,p)$ and $\rho+\delta F(y;\alpha,p)$ have at most two points of intersection.}
		\label{Fig-2}	
	\end{figure}
	
	Based on the shapes of the Gamma law,
	in Figure \ref{Fig-2}, we graphically sketch the functions $2\delta f(y;\alpha,p)$ and $\rho+\delta F(y;\alpha,p)$, and consider all possible cases, by varying the parameters  $\alpha$, $\rho$, $\delta$ and $p$ (known), in which the graphs of $2\delta f(y;\alpha,p)$ and $\rho+\delta F(y;\alpha,p)$ intersect (or not).
	
	So, for the cases when $p\leqslant 1$ or  $p> 1$, the function $R$ in \eqref{def-R-f} has at most two real zeros. This completes the proof.
\end{proof}

\section{Moments in the Gaussian case}\label{Moments in the Gaussian case}
In this section, for simplicity, we consider $p\geqslant 1$ an integer.
To get the moments of $X\sim {\rm TD}_\Phi(\boldsymbol{\theta})$,
by Corollary \ref{moments}, it is necessary to determine 
$
\mathbbm{E}(\mathds{1}_{\{Z\leqslant\pm\alpha\sqrt{Y}\}} Z^n)
$, $n\geqslant 1$,
where $Z\sim N(0,1)$ and $Y\sim{\rm Gamma}(p,1)$ are independent. Indeed, since $Z$ and $Y$ are independent, by definition of expectation,
\begin{align}\label{ide-expc}
\mathbbm{E}\big(\mathds{1}_{\{Z\leqslant\pm\alpha\sqrt{Y}\}} Z^n\big)
=
\int_{0}^{\infty} 
\mathbbm{E}\big(\mathds{1}_{\{Z\leqslant\pm\alpha\sqrt{y}\}} Z^n\big)
\, 
{y^{p-1} {\rm e}^{-y}\over\Gamma(p)}
\, {\rm d}y.
\end{align}

Let us define $n_m=n-2m-1$,
\begin{align}
c_{n,m}= {n!2^{-m} n_m!\over m!(n_m+1)!}
\quad \text{and}\quad
d_{n_m,j}= {(-1)^j\over j!(n_m-2j)! 2^j}. \label{exp-tri-G}
\end{align}
By using the following known formula (see Theorem 2.3. of \cite{PENDER201540}): for each real numbers $a,b$ such that $a<b$,
\begin{align*}
&\mathbbm{E}\big(\mathds{1}_{\{a\leqslant Z\leqslant b\}} Z^n\big)
=
\sum_{m=0}^{\lfloor n/2\rfloor}
\sum_{j=0}^{\lfloor n_m/2\rfloor}
c_{n,m} d_{n_m,j}\,
\big[a^{n_m-2j}\phi(a)-b^{n_m-2j}\phi(b)\big],
\end{align*}
where ${\lfloor x\rfloor}$ is the greatest integer less than or equal to $x$, we have
\begin{align*}
\mathbbm{E}\big(\mathds{1}_{\{Z\leqslant\pm\alpha\sqrt{y}\}} Z^n\big)
=
-
\sum_{m=0}^{\lfloor n/2\rfloor}
\sum_{j=0}^{\lfloor n_m/2\rfloor}
c_{n,m} d_{n_m,j}\,
(\pm\alpha\sqrt{y})^{n_m-2j}\phi(\pm\alpha\sqrt{y}).
\end{align*}
Replacing the above identity in  \eqref{ide-expc},
\begin{align*}
\mathbbm{E}\big(\mathds{1}_{\{Z\leqslant\pm\alpha\sqrt{Y}\}} Z^n\big)
=
-
\sum_{m=0}^{\lfloor n/2\rfloor}
\sum_{j=0}^{\lfloor n_m/2\rfloor}
c_{n,m} d_{n_m,j}\,
\int_{0}^{\infty} 
(\pm\alpha\sqrt{y})^{n_m-2j} \phi(\pm\alpha\sqrt{y})
\, 
{y^{p-1} {\rm e}^{-y}\over\Gamma(p)}
\, {\rm d}y.
\end{align*}
But, by definition of gamma distribution, 
\begin{align*}
\int_{0}^{\infty} 
(\pm\alpha\sqrt{y})^{n_m-2j} \phi(\pm\alpha\sqrt{y})
\, 
{y^{p-1} {\rm e}^{-y}\over\Gamma(p)}
\, {\rm d}y
=
{(\pm\alpha)^{n_m-2j}\over\sqrt{2\pi}\Gamma(p)(1+{\alpha^2\over 2})^{{n_m-2j\over 2}+p}}\,
\Gamma\biggl({n_m-2j\over 2}+p\biggr).
\end{align*}
Hence
\begin{eqnarray}\label{exp-z-truncated}
\mathbbm{E}\big(\mathds{1}_{\{Z\leqslant\pm\alpha\sqrt{Y}\}} Z^n\big)
=
-{1\over \sqrt{2\pi}\Gamma(p)}
\sum_{m=0}^{\lfloor n/2\rfloor}
\sum_{j=0}^{\lfloor n_m/2\rfloor}
{c_{n,m} d_{n_m,j}\,(\pm\alpha)^{n_m-2j}\over(1+{\alpha^2\over 2})^{{n_m-2j\over 2}+p}}\,
\Gamma\biggl({n_m-2j\over 2}+p\biggr).
\end{eqnarray}

By using \eqref{exp-z-truncated} and the known formula
\begin{align*}
\mathbbm{E}(Z^n)
={2^{-n/2} n!\over (n/2)!}\, \mathds{1}_{\{n \, {\rm even}\}},
\end{align*}
in Corollary \ref{moments}, we get the formula of moments for $X\sim {\rm TD}_\Phi(\boldsymbol{\theta})$ given in  Subsection \ref{Moments and Shannon entropy}.

\section{Entropy in the Gaussian case}\label{Entropy in the Gaussian case}
By Remark \ref{rem-H1}, the Shannon entropy $S_1(X)$ of $X\sim {\rm TD}_\Phi(\boldsymbol{\theta})$ exists whenever $S_1(Z)$ of  $Z\sim N(0,1)$ also exists. Since $S_1(Z)=\log(\sqrt{2\pi})+{1/ 2}$, the  existence of $S_1(X)$ is guaranteed. So it makes sense to find a closed expression for $S_1(X)$.
To find this expression,
by Proposition \ref{prop-Shannon entropy}, it is enough to determine the expectations:
$
\mathbbm{E}[\mathds{1}_{\{Z\leqslant \pm\alpha\sqrt{Y}\}} \log(\rho+\delta T(Z;\alpha,p))]
$
and
$\mathbbm{E}[S_1(\mathds{1}_{\{Z\leqslant \pm\alpha\sqrt{Y}\}}Z)]
$,
where $T$ is as in \eqref{tranf-G} and $Z\sim N(0,1)$ and $Y\sim{\rm Gamma}(p,1)$ are independent. As in the previous subsection, for simplicity, in this subsection we consider $p\geqslant 1$ an integer.

A simple calculation shows that, for each real numbers $a,b$ such that $a<b$,
\begin{align*}
S_1(\mathds{1}_{\{a\leqslant Z\leqslant b\}}Z)
&=
-\int_{a}^b \phi(z) \log\phi(z) \, {\rm d}z
\\[0,2cm]
&=
{1\over 2}\,
\big[a\phi(a)-b\phi(b)\big]
+
{1\over 2}\,
\biggl[\log(\sqrt{2\pi})+{1\over 2}\biggr]
\left[{\rm erf}\biggl({b\over \sqrt{2}}\biggr) 
- {\rm erf}\biggl({a\over \sqrt{2}}\biggr) \right].
\end{align*}
Then
\begin{align*}
\mathbbm{E}\big[S_1(\mathds{1}_{\{Z\leqslant \pm\alpha\sqrt{Y}\}}Z)\big]
&=
\mp
{1\over 2}
\int_{0}^{\infty}
\alpha\sqrt{y}\phi(\pm\alpha\sqrt{y})\,
{y^{p-1} {\rm e}^{-y}\over\Gamma(p)}
\, {\rm d}y
\\[0,2cm]
&+
{1\over 2}
\biggl[\log(\sqrt{2\pi})+{1\over 2}\biggr]
\int_{0}^{\infty}
\left[{\rm erf}\biggl({\pm\alpha\sqrt{y}\over \sqrt{2}}\biggr) 
+1 \right]
{y^{p-1} {\rm e}^{-y}\over\Gamma(p)}
\, {\rm d}y.
\end{align*}
But,
\begin{itemize}
	\item  by definition of Gamma distribution,
	\begin{align*}
	\int_{0}^{\infty}
	\alpha\sqrt{y}\phi(\pm\alpha\sqrt{y})\,
	{y^{p-1} {\rm e}^{-y}\over\Gamma(p)}
	\, {\rm d}y
	=
	{\alpha\Gamma(p+{1\over 2})\over \sqrt{2\pi}\Gamma(p)}\,
	{1\over (1+{\alpha^2\over 2})^{p+{1\over 2}}},
	\end{align*}
	\noindent
	\item
	and by using \eqref{exp-Phi-0}, 
	\begin{align*}
	\int_{0}^{\infty}
	\left[{\rm erf}\biggl({\pm\alpha\sqrt{y}\over \sqrt{2}}\biggr) 
	+1 \right]
	{y^{p-1} {\rm e}^{-y}\over\Gamma(p)}
	\, {\rm d}y
	=
	2{\mathbbm{E}}
	\big[
	\Phi(\pm\alpha \sqrt{Y})
	\big].
	\end{align*}
\end{itemize}
Hence
\begin{eqnarray}\label{exp-form-1}
\mathbbm{E}\big[S_1(\mathds{1}_{\{Z\leqslant \pm\alpha\sqrt{Y}\}}Z)\big]
=
\mp
{\alpha\Gamma(p+{1\over 2})\over 2\sqrt{2\pi}\Gamma(p)}\,
{1\over (1+{\alpha^2\over 2})^{p+{1\over 2}}} 
+
\biggl[\log(\sqrt{2\pi})+{1\over 2}\biggr]
{\mathbbm{E}}
\big[
\Phi(\pm\alpha \sqrt{Y})
\big].
\end{eqnarray}

The proof of the following result is technical.
\begin{proposition} 
	If $Z\sim N(0,1)$ the following hold
	\begin{align}\label{exp-log-0}
	&\mathbbm{E}\big[ \log(\rho+\delta T(Z;\alpha,p))\big]
	=
	2\sum_{k=0}^{\infty} {1\over (2k+1)} 
	\biggl[
	1
	+
	{1\over (\rho+1)^{2k+1}} 
	\sum_{i=p}^{\infty}
	\widetilde{c}_{i,k}\,
	{2^{-i} (2i)!\over i!}
	\biggr];
	\\[0,3cm]
	&\mathbbm{E}\big[\mathds{1}_{\{Z\leqslant \pm\alpha\sqrt{Y}\}} \log(\rho+\delta T(Z;\alpha,p))\big] \nonumber
	\\[0,2cm]
	&\hspace*{2.7cm}
	=
	2\sum_{k=0}^{\infty} {1\over (2k+1)} 
	\biggl\{
	\mathbbm{E}\big[\Phi(\pm\alpha\sqrt{Y})\big]
	+
	{1\over (\rho+1)^{2k+1}} 
	\sum_{i=p}^{\infty}
	\widetilde{c}_{i,k} \mathbbm{E}\big(\mathds{1}_{\{Z\leqslant \pm\alpha\sqrt{Y}\}} {Z^{2i}}\big)
	\biggr\}.  \label{exp-log}
	\end{align}
\end{proposition}
\begin{proof}
We consider the power series expansions of $\log(z)$ and of the incomplete gamma function:
\begin{align}\label{series expansion}
\log(z)
=
2\sum_{k=0}^{\infty} {1\over 2k+1} \biggl({z-1\over z+1}\biggr)^{2k+1},
\quad 
\gamma(p,z)
=
\sum_{k=0}^{\infty}
{(-1)^k\over k!}\, {z^{p+k}\over p+k}
=
\sum_{k=p}^{\infty}
{(-1)^{k-p}\over (k-p)!}\, {z^{k}\over k},
\end{align}
respectively.
Using these expansions, we have
\begin{align*}
\log(\rho+\delta T(z;\alpha,p))
&=
2\sum_{k=0}^{\infty} {1\over 2k+1} 
\biggl[{\rho-1+\delta T(z;\alpha,p)\over \rho+1+\delta T(z;\alpha,p)}\biggr]^{2k+1}
\\[0,2cm]
&=
2\sum_{k=0}^{\infty} {1\over 2k+1} 
\Biggl[{\rho-1+{\delta\over\Gamma(p)} 
	\sum_{i=p}^{\infty}
	{(-1)^{i-p}\over 2i(i-p)! \alpha^{2i}}\, {z^{2i}}
	\over \rho+1+{\delta \over\Gamma(p)}
	\sum_{j=p}^{\infty}
	{(-1)^{j-p}\over 2j(j-p)! \alpha^{2j}}\, {z^{2j}}
}\Biggr]^{2k+1}.
\end{align*}
Letting 
$b_0=\rho-1$ and
$b_i=[{\delta(-1)^{i-p}/ 2i(i-p)! \alpha^{2i} \Gamma(p)}]\mathds{1}_{\{i\geqslant p\}}$; 
and similarly, defining
$a_0=\rho+1$ and
$a_j=[{\delta(-1)^{j-p}/ 2j(j-p)! \alpha^{2j}}\Gamma(p)]\mathds{1}_{\{j\geqslant p\}}$;
the last series can be expressed as
\begin{align*}
=
2\sum_{k=0}^{\infty} {1\over 2k+1} 
\Biggl(
{
	\sum_{i=0}^{\infty}
	b_i {z^{2i}}
	\over 
	\sum_{j=0}^{\infty}
	a_j {z^{2j}}
}\Biggr)^{2k+1}.
\end{align*}
By application of the equation in Section 0.313 of \cite{Gradshteyn2000} for division of power series, the above series is written as
\begin{align}\label{rec-form-1}
=
2\sum_{k=0}^{\infty} {1\over 2k+1} 
\Biggl(
{1\over \rho+1}
\sum_{i=0}^{\infty}
c_i {z^{2i}}
\Biggr)^{2k+1},
\end{align}
where the coefficients $c_i$'s are determined
from the recurrence equation
\begin{align*}
c_n=\biggl(b_n-{1\over \rho+1}\sum_{i=p}^{n} c_{n-i} a_i\biggr)\mathds{1}_{\{n\geqslant p\}}.
\end{align*}
By application of the equation in Section 0.314 of \cite{Gradshteyn2000}
for power series raised to powers, the series \eqref{rec-form-1} is equal to
\begin{align*}
2\sum_{k=0}^{\infty} {1\over (2k+1)(\rho+1)^{2k+1}} 
\sum_{i=0}^{\infty}
\widetilde{c}_{i,k} {z^{2i}}
=
2\sum_{k=0}^{\infty} {1\over (2k+1)} 
\biggl[
1
+
{1\over (\rho+1)^{2k+1}} 
\sum_{i=p}^{\infty}
\widetilde{c}_{i,k} {z^{2i}}
\biggr],
\end{align*}
where the coefficients $\widetilde{c}_{i,k}$'s are determined by
$\widetilde{c}_{0,k}=(\rho+1)^{2k+1}$
and from the recurrence
relation 
\begin{align}\label{c-srec}
\widetilde{c}_{m,k}=
{1\over m(\rho+1)}
\sum_{i=p}^{m}\big[2i(k+1)-m\big]a_i \widetilde{c}_{m-i,k} \mathds{1}_{\{m\geqslant p\}}.
\end{align}
In short, we have 
\begin{align*}
\log(\rho+\delta T(z;\alpha,p))
=
2\sum_{k=0}^{\infty} {1\over (2k+1)} 
\biggl[
1
+
{1\over (\rho+1)^{2k+1}} 
\sum_{i=p}^{\infty}
\widetilde{c}_{i,k} {z^{2i}}
\biggr].
\end{align*}

By using the above expansion and assuming we can change the series with the integral sign, we have
	\begin{align*}
	\mathbbm{E}\big[\mathds{1}_{\{a\leqslant Z\leqslant b\}} \log(\rho+\delta T(Z;\alpha,p))\big]
	&=
	\int_{a}^{b}
	\log(\rho+\delta T(z;\alpha,p))
	\phi(z)\, {\rm d}z \nonumber
	\\[0,2cm]
	&=
	2\sum_{k=0}^{\infty} {1\over (2k+1)} 
	\biggl[
	\Phi(b)-\Phi(a)
	+
	{1\over (\rho+1)^{2k+1}} 
	\sum_{i=p}^{\infty}
	\widetilde{c}_{i,k} \mathbbm{E}\big(\mathds{1}_{\{a\leqslant Z\leqslant b\}} {Z^{2i}}\big)
	\biggr].   
	\end{align*}
From the above identity the identities in \eqref{exp-log-0} and \eqref{exp-log} follow.
\end{proof}

Since $S_1(Z)=\log(\sqrt{2\pi})+{1/ 2}$, by substituting  identities \eqref{exp-log-0}, \eqref{exp-log} and \eqref{exp-form-1}  in Proposition \ref{prop-Shannon entropy}, we get the formula of the Shannon entropy for $X\sim {\rm TD}_\Phi(\boldsymbol{\theta})$ given in Subsection \ref{Moments and Shannon entropy}.

\section{Derivatives of the function
\texorpdfstring{$\log f(x;\boldsymbol{\theta})$}{k} }
\label{Elements of the score vector}
Let  $Y\sim{\rm Gamma}(p,1)$ and let $T$ be as in \eqref{tranf-G}. Then

\begin{align*}
	\begin{array}{lllll}
	\displaystyle
		{\partial \log f(x;\boldsymbol{\theta})\over\partial \mu} 
		= 
		-\frac{2 (\frac{x-\mu}{\sigma \alpha} )^{2(p-1)}\,
			{\rm e}^{ - (\frac{x-\mu}{\alpha\sigma})^2 }}{\sigma^2 \alpha^2\Gamma(p)}\,
		\log\left[
		\rho + \delta T\left( \frac{x-\mu}{\sigma}; \alpha,p \right) \right] 
		  -
		  {1\over\sigma} 
		  \frac{g' \left( \frac{x-\mu}{\sigma} \right)}{g \left( \frac{x-\mu}{\sigma} \right)};		
		\\[0,7cm]
		\displaystyle
		{\partial \log f(x;\boldsymbol{\theta})\over\partial \sigma}
		= 
		-  
		\frac{\rho + \delta}{Z_{\boldsymbol \theta}}
		-
		\frac{\delta}{Z_{\boldsymbol \theta}}
		\big\{ \mathbbm{E}[G(-\alpha \sqrt{Y})] - \mathbbm{E}[G(\alpha \sqrt{Y} ) ]\big\}  
				\\[0,6cm]
		\displaystyle
		\hspace*{8.61cm}
		-
		\frac{2\delta ( \frac{x-\mu}{\sigma \alpha}  )^{2p-1} 
		{\rm e}^{ - (\frac{x-\mu}{\alpha\sigma})^2 }
	}{\sigma^3\left[\rho + \delta T(\frac{x-\mu}{\sigma};\alpha;p)\right] \Gamma(p)} 
		-
\frac{(x-\mu)}{\sigma^2} 
\frac{g' \left( \frac{x-\mu}{\sigma} \right)}{g \left( \frac{x-\mu}{\sigma} \right)}  ;
		\\[0,7cm]
		\displaystyle
		{\partial \log f(x;\boldsymbol{\theta})\over\partial \alpha}
		= 
		-
		\frac{2 \sigma \delta\Gamma(p+1/2)}{\sqrt{2\pi} Z_{\boldsymbol{\theta}}\Gamma(p)}\,
		{_2F_1}\biggl( 
		\frac{1}{2}; p+\frac{1}{2}; \frac{3}{2}; \frac{-\alpha^2}{2} 
		\biggr) 
		\\[0,6cm]
		\displaystyle
		\hspace*{2.3cm}
		- 
		\frac{2 \sigma \Gamma(p+1/2)}{\sqrt{2\pi}Z_{\boldsymbol{\theta}}\Gamma(p)}
		\biggl[ 
		\biggl(
		\frac{\alpha^2}{2} + 1 
		\biggr)^{-(p+{1\over 2})} 
		\!\!\!- 
		{_2F_1}
		\biggl(
		\frac{1}{2};p+\frac{1}{2};\frac{3}{2};\frac{-\alpha^2}{2}
		\biggr) 
		\biggr] 
		\!-\! 
		\frac{2\delta 
		(\frac{x-\mu}{\sigma \alpha} )^{2p}\, 
	{\rm e}^{- (\frac{x-\mu}{\alpha\sigma})^2 }	
	}{\alpha\left[\rho + \delta T( \frac{x-\mu}{\sigma}; \alpha, p)\right]};
		\\[0,7cm]
		\displaystyle
		{\partial \log f(x;\boldsymbol{\theta})\over\partial \rho}
		= 
		-\frac{\sigma}{Z_{\boldsymbol \theta}} +  \frac{1}{\rho + \delta T(\frac{x-\mu}{\sigma};\alpha;p)};
		\\[0,7cm]
		\displaystyle
		{\partial \log f(x;\boldsymbol{\theta})\over\partial \delta}
		= 
		-
		\frac{\delta}{Z_{\boldsymbol \theta}} 
		\big\{ 
		1 +  \mathbbm{E}[G(-\alpha\sqrt{Y})] - \mathbbm{E}[G(\alpha \sqrt{Y})]     
		\big\} 
		+ 
		\frac{T\left( \frac{x-\mu}{\sigma}; \alpha, p \right)}{ \rho + \delta T \left( \frac{x-\mu}{\sigma}; \alpha, p \right)    }.
	\end{array}
\end{align*}

\section{Inference, Optimization, Bootstrap and Applications}


\subsection{Rao-Cramer lower bounds of estimators from ML$_q$E and MLE}

Since we provide the square roots of Rao-Cramer lower bounds of estimators from ML$_q$E and MLE, the confidence interval of estimators can be constructed as the following form:
\begin{equation*}
	\hat{\theta} \mp z_{\tau/2} \sqrt{\text{Var}_q(\hat{\theta})},
\end{equation*}
\noindent where $\hat{\theta}$ is chosen as $(\hat{\mu},\hat{\sigma},\hat{\alpha},\hat{\rho},\hat{\delta})$. $z_{\tau/2}$ is the critical value from the standard normal distribution when the significance level $\tau$ is chosen by researcher \cite{hub64,Lehmann1998,Ferrari2010}. $\text{Var}_q(\hat{\theta})$ is the theoretical variance of estimator evaluated numerically by the codes in Section \ref{cisect}. Thus, one can observe the limit values of estimates of parameters to generate the different kind of artificial data sets if it is necessary to do so at an experiment.

The lower limit values of confidence intervals for the estimates must be as follows: $\hat{\sigma}_L>0$, $\hat{\alpha}_L>0$, $\hat{\rho}_L \geqslant 0$ and $\hat{\delta}_L \geqslant 0$. In other words, these values should be greater than zero due to the defined values of these parameters.

\subsection{Benefit of bootstrapped data when semiparametric and parametric models are used for fitting data sets}\label{benefitof}

The general tendency for modelling is performed by the existing techniques which are the estimated distribution including parametric models, kernel smoothing. As is expected, the kernel smoothing technique as a semiparametric approach for modelling can show better performance for finite sample size when compared with parametric models \cite{smoothker}. However, we can need to assume that a data set is member of a parametric model. In this case, if an artificial data set is necessary to model the probable results in the future or if conducting an experiment is expensive for researchers, the necessity of parametric models such as bimodal \cite{denizspain} and trimodal normal distribution can be inevitable; because, the smooth kernel distribution used for fitting on the data set is data-adaptive or semiparametric approach \cite{yellowcoversmooth}. In this case, the parametric model puts a restriction in where we have situations which give the results mimicing the parametric model perfectly because of the probabilities coming from PDF. If we have such results, then we need to make a comparison between kernel smoothing and the trimodal forms of a function which can be chosen as normal, Student $t$, etc. \cite{bootmathe}. 

Since the real data set includes results taken at the moment or sometimes experiments have to be conducted at only one time due to the cost of running of experiment, different collections of the same data should be generated by use of the bootstrap technique.  Any mixed models as a hetero form can be modelled by using  TD$_\Phi$  distribution if the trimodality exists. Since hetero-mixing forms a non-identically distributed data set, trimodal distribution can be necessary to perform an efficient fitting on the data sets observed from an experiment. 

Since bootstrap technique performs a random choice from a data set of experiment, we need to make a comparison among the semiparametric and parametric models whether or not which one will be performable in fitting on the data set \cite{bootmathe,yellowcoversmooth}.  The soft and strict forms of trimodal normal distributions are capable to perform an efficient fitting on data set and is an alternative approach when compared with smoothing technique if trimodality exists. Because, it is concluded that a parametric model must be necessary to generate the random numbers from corresponding probabilities of parametric model, i.e PDF, for research instead of doing an another experiment again. If we generate artificial numbers by use of smoothing technique, these numbers will be generated according to the used smoothing technique which has its probabilities corresponding empirical probability function representing the probabilities coming from each event in the finite sample size (see Subsection \ref{sktechapp}) \cite{yellowcoversmooth}. When we come across big data analytics, the population starts to evolve and take what its real form is as. In this case, a parametric model can be necessary \cite{bigdata}.

 
 
 

\subsection{Driving the bootstrap and the optimization of $\log_q$ likelihood function for estimations of parameters}

GOFTs which perform a testing for assesment of the used (semi)-parametric model can show different performance for each model, which is why the different GOFTs have been applied to test the fitting performance of models. The bootstrap technique in Mathematica  is applied to the real data set (see Appendix \ref{bootsim}). The replication number of bootstrap is $1000$. Thus, the probable mistakings in convergence of the optimization and the different scenarios of real data sets are tried  to be clarified. It is also noted that the data observed after the experiment is conducted can include the measurement error and the measurements can depend on many factors which are known as the randomness, mistakes, unhidden factors, etc. in the experiment \cite{bootmathe,quantummathe}. 

The bootstrapped form of the real data is not only beneficial to jump on these kind of problems in the process of measurement but also the replicated optimization should be useful for us to make the convergence around global point (GP) or reach the real GP where the optimization can reach if the replication is performed. Further, SK technique depends on the chosen function for kernel and it is a semiparametric technique \cite{smoothker,yellowcoversmooth}. As is shown by Table \ref{tableex1}, SK has performance on the modelling. However, the performance of trimodal normal distribution (TD$_\Phi$) can have superior performance when the closeness to statistics generated by SK and the estimated distribution (EstD) which is parametric method is taken into account. EstD integrated into Mathematica software tries to find the best function (mixed form of the parametric models or other parametric models)  while conducting the fitting on the real data set. For example, a data set can be a mixing form of two or three normal distributions with the corresponding parameters  from populations $1$ $2$ or $3$, respectively. If we have a mixed form for the real data set, it can also be member of trimodal distribution which is eventually not known by researchers, which is an important gap in conducting a research. For this reason, TD$_\Phi$ has been proposed and also we make a comparison among them to observe the performance of parametric model (see also discussion in introduction Section \ref{introsec}).

The initial values for starting the optimization performed by \texttt{FindMaximum} with constraints in the parameters $\alpha$, $\rho$ and $\delta$ searched at the interval $(0,100]$ are generated by uniform distribution with $[0,1]$ \cite{bootmathe,quantummathe}. 

%
%

\section{Mathematica 12.0 codes}
\subsection{Generating procedure for random numbers if \texttt{SmoothKernelDistribution} is used}\label{sktechapp}
The codes for generating random numbers from the smooth kernel distribution in Mathematica is given by the following lines:

\begin{verbatim}
	SK = SmoothKernelDistribution[x, {"Adaptive", Automatic, .1},PerformanceGoal
	-> "Quality"];
	F[m_, s_] := CDF[SK, x];
	m := Moment[SK, 1]; s := Sqrt[Moment[SK, 2] - Moment[SK, 1]^2];
	For[i = 0, i < n, i++,
	z = Table[
	x /. FindRoot[F[m, s] - RandomReal[], {x, Lower point, Upper point}], {i, n}]
	]
\end{verbatim}

One can get the random number from trimodal normal distribution by use of  subsection \ref{stocrep} or CDF  of TD$_\Phi$ in Subsection \ref{cumfunTD}.


\subsection{Optimization and statistics}
\begin{verbatim}
iv := RandomReal[];
FM = FindMaximum[{Total[Log[f]], 
	a< \[Alpha] <= b && a <= \[Rho] <= b && 
	a <= \[Delta] <= b}, {{\[Alpha], iv}, {\[Rho] , 
		iv}, {\[Delta], iv}, {\[Mu], Median[x]}, {\[Sigma], 
		Median[Abs[x - Median[x]]]}}];
	EstD = FindDistribution[x]; 
	SK = SmoothKernelDistribution[x, {"Adaptive", Automatic, .1},PerformanceGoal
	 -> "Quality"];
 statistics = 
 {
 	Moment[EstD, 1], Moment[SK, 1],  Mean[x], Median[x],  
	  Sqrt[Moment[EstD, 2] - Moment[EstD, 1]^2], 
	  Sqrt[Moment[SK, 2] - Moment[SK, 1]^2], 
	  StandardDeviation[x], Median[Abs[x - Median[x]]]
 };
	
\end{verbatim}
\subsection{The codes for bootstrap}\label{bootsim}
A part for bootstrap is given by the following form \cite{bootmathe}:
\begin{verbatim}
	For[i = 1, i <= replication, i++,
	     	x := Table[RandomChoice[data, Length[data]], {i,1,Length[data]}]
	   ]
\end{verbatim}

\subsection{CDF of TN$_\Phi$}
\begin{verbatim}
p := 3/2; ns := n;
x = Sort[x];
nsn = 0; nsp = 0;
For[i = 1, i <= ns, i++,
If[x[[i]] < \[Mu], nsn = nsn + 1]
]
For[i = 1, i <= nsn, i++,
If[x[[i]] < \[Mu], xn[i] := x[[i]]]
]
For[i = nsn + 1, i <= ns, i++,
If[x[[i]] >= \[Mu], xp[i] := x[[i]]]
]
For[i = 1, i <= nsn, i++,
FN[i] = (\[Sigma]/
Z)*(\[Rho] + \[Delta]*
T[xn[i]/\[Sigma] - \[Mu]/\[Sigma], \[Alpha], p])*
CDF[NormalDistribution[\[Mu], \[Sigma]], 
xn[i]] + ((\[Delta]*\[Sigma])/
Z)*(1/2 - (\[Alpha]*Gamma[p + 1/2]*
Hypergeometric2F1[1/2, p + 1/2, 3/2, -\[Alpha]^2/2])/(Sqrt[
2*Pi]*Gamma[p])) - ((\[Delta]*\[Sigma])/
Z)*((1/2)*T[xn[i]/\[Sigma] - \[Mu]/\[Sigma], \[Alpha], p] - 
NIntegrate[
Erf[(\[Alpha]/Sqrt[2])*z]*z^(2*p - 1)*Exp[-1^2*z^2], {z, 
	0, (xn[i] - \[Mu])/(\[Alpha]*\[Sigma])}, 
Method -> {"LobattoKronrodRule"}]/Gamma[p]);
]
For[i = nsn + 1, i <= ns, i++,
FP[i] = (\[Sigma]/
Z)*(\[Rho] + \[Delta]*
T[xp[i]/\[Sigma] - \[Mu]/\[Sigma], \[Alpha], p])*
CDF[NormalDistribution[\[Mu], \[Sigma]], 
xp[i]] + ((\[Delta]*\[Sigma])/
Z)*(1/2 - (\[Alpha]*Gamma[p + 1/2]*
Hypergeometric2F1[1/2, p + 1/2, 3/2, -\[Alpha]^2/2])/(Sqrt[
2*Pi]*Gamma[p])) - ((\[Delta]*\[Sigma])/
Z)*((1/2)*T[xp[i]/\[Sigma] - \[Mu]/\[Sigma], \[Alpha], p] + 
NIntegrate[
Erf[(\[Alpha]/Sqrt[2])*z]*z^(2*p - 1)*Exp[-1^2*z^2], {z, 
	0, (xp[i] - \[Mu])/(\[Alpha]*\[Sigma])}, 
Method -> {"LobattoKronrodRule"}]/Gamma[p]);
]
FNFP := Join[{Table[FN[j], {j, nsn}], Table[FP[j], {j, nsn + 1, ns}]}];
JoinFNP := Join[FNFP[[1]], FNFP[[2]]];
\end{verbatim}
\subsection{Goodness of fit tests for the proposed and used distributions}
The goodness of fit tests such as  Kolmogorov-Smirnov (KS), Cram\'{e}r--von Mises and Anderson-Darling are evaluated by using the the following codes for TN$_\Phi$, EstD and SK. 
\begin{verbatim}
For[i = 1, i <= ns, i++,
CVM[i] = (JoinFNP[[i]] - (2*i - 1)/(2*ns))^2 + 1/(12*ns);
AD[i] = (-1/
ns)*((2*i - 1)*(Log[JoinFNP[[i]]] + Log[1 - JoinFNP[[i]]]));
DP[i] = i/ns - JoinFNP[[i]]; DN[i] = JoinFNP[[i]] - (i - 1)/ns;
CVMED[i] = (CDF[estimated\[ScriptCapitalD], x][[
i]] - (2*i - 1)/(2*ns))^2 + 1/(12*ns);
ADED[i] = (-1/
ns)*((2*i - 1)*(Log[CDF[EstD], x][[i]]] + 
Log[1 - CDF[EstD], x][[i]]]));
DPED[i] = i/ns - CDF[EstD], x][[i]]; 
DNED[i] = CDF[estimated\[ScriptCapitalD], x][[i]] - (i - 1)/ns;
CVMSK[i] = (CDF[SK, x][[i]] - (2*i - 1)/(2*ns))^2 + 1/(12*ns);
ADSK[i] = (-1/ns)*((2*i - 1)*(Log[CDF[SK, x][[i]]] + Log[1 - CDF[SK, x][[i]]]));
DPSK[i] = i/ns - CDF[SK, x][[i]]; 
DNSK[i] = CDF[SK, x][[i]] - (i - 1)/ns;
]
DPmax = Max[Table[DP[j], {j, ns}]]; DNmax = Max[Table[DN[j], {j, ns}]]; 
KS = Max[DPmax, DNmax];StaCVM = Total[Table[CVM[j], {j, ns}]];
StaAD = Total[Table[AD[j], {j, ns}]];{KS, StaCVM, StaAD}
\end{verbatim}

In order to get information criteria (IC) such as Akaike and Bayesian,  \texttt{CDF} of Mathematica must be replaced with its corresponding \texttt{PDF} and the formulae of IC are used. For example,
\begin{verbatim}
	 -2*Total[Log[PDF[SK, x]]] + 2*p
\end{verbatim}

\subsection{The codes for variances of estimators}\label{cisect}
The PDF of TD$_\Phi$ is given by
\begin{verbatim}
f[\[Mu]_, \[Sigma]_, \[Alpha]_, \[Rho]_, \[Delta]_] := \
\[Sigma]*((1/((\[Rho] + \[Delta])*\[Sigma] - ((Sqrt[
2]*\[Delta]*\[Sigma]*\[Alpha]*
Gamma[p + 1/2])/(Sqrt[Pi]*Gamma[p]))*
Hypergeometric2F1[1/2, p + 1/2, 
3/2, -\[Alpha]^2/2]))*(\[Rho] + \[Delta]*(1 - 
Gamma[p, (x - \[Mu])^2/(\[Sigma]^2* \[Alpha]^2)]/
Gamma[p]))*(exp^(-((x - \[Mu])^2/(2 \[Sigma]^2)))/(
Sqrt[2 \[Pi]] \[Sigma])));
g[\[Mu]_, \[Sigma]_] := exp^(-((x - \[Mu])^2/(2 \[Sigma]^2)))/(
Sqrt[2 \[Pi]] \[Sigma]); 
T[p_, \[Mu]_, \[Sigma]_, \[Alpha]_] := 
1 - Gamma[p, (x - \[Mu])^2/(\[Sigma]^2* \[Alpha]^2)]/Gamma[p];

Dgm2[\[Mu], \[Sigma]] := D[g[\[Mu], \[Sigma]], \[Mu]]; 

Z := (\[Rho] + \[Delta])*\[Sigma] - ((Sqrt[
2]*\[Delta]*\[Sigma]*\[Alpha]*Gamma[p + 1/2])/(Sqrt[Pi]*
Gamma[p]))*
Hypergeometric2F1[1/2, p + 1/2, 3/2, -\[Alpha]^2/2];
\end{verbatim}
The score ${\partial \log  f(x;\boldsymbol{\theta}) \over \partial \mu}$ for $\mu$ is given by
\begin{verbatim}
em := ((-2*((x - \[Mu])/(\[Sigma]*\[Alpha]))^(2*p - 2)*
Exp[-((x - \[Mu])/(\[Sigma]*\[Alpha]))^2])/(\[Sigma]^2*\
\[Alpha]^2*Gamma[p]))*
Log[\[Rho] + \[Delta]*
T[p, \[Mu], \[Sigma], \[Alpha]]] - (1/\[Sigma])*
Dgm2[\[Mu], \[Sigma]];
\end{verbatim}
The score ${\partial \log  f(x;\boldsymbol{\theta}) \over \partial \sigma}$ for $\sigma$ is given by
\begin{verbatim}
es := -\[Rho]/Z - \[Delta]/
Z - (\[Delta]/
Z)*(NIntegrate[
CDF[NormalDistribution[\[Mu], \[Sigma]], -\[Alpha]*y^0.5]*
PDF[GammaDistribution[p, 1], y], {y, 0, Infinity}, 
Method -> {"LobattoKronrodRule"}] - 
NIntegrate[
CDF[NormalDistribution[\[Mu], \[Sigma]], \[Alpha]*y^0.5]*
PDF[GammaDistribution[p, 1], y], {y, 0, Infinity}, 
Method -> {"LobattoKronrodRule"}]) - ((2*\[Delta]*((x - \
\[Mu])/(\[Sigma]*\[Alpha]))^(2*p - 1)*
Exp[-((x - \[Mu])/(\[Sigma]*\[Alpha]))^2])/(\[Sigma]^3*(\[Rho] \
+ \[Delta]*T[p, \[Mu], \[Sigma], \[Alpha]])*
Gamma[p])) - (((x - \[Mu])/\[Sigma]^2)*Dgs2[\[Mu], \[Sigma]]);
\end{verbatim}
The score ${\partial \log  f(x;\boldsymbol{\theta}) \over \partial \alpha}$ for $\alpha$ is given by
\begin{verbatim}
ea := -(2*\[Sigma]*\[Delta]*Gamma[p + 0.5])/(Sqrt[2*Pi]*Z*Gamma[p])*
Hypergeometric2F1[1/2, p + 1/2, 
3/2, -\[Alpha]^2/
2] - (2*\[Sigma]*Gamma[p + 1/2])/(Sqrt[2*Pi]*Z*
Gamma[p])*((\[Alpha]^2/2 + 1)^(-p - 1/2) - 
Hypergeometric2F1[1/2, p + 1/2, 
3/2, -\[Alpha]^2/
2]) - (2*\[Delta]*((x - \[Mu])/(\[Sigma]*\[Alpha]))^(2*p)*
Exp[-((x - \[Mu])/(\[Sigma]*\[Alpha]))^2])/(\[Alpha]*(\[Rho] + \
\[Delta]*T[p, \[Mu], \[Sigma], \[Alpha]]));
NIntegrate[
em^2*f[\[Mu], \[Sigma], \[Alpha], \[Rho], \[Delta]]^(2 - 
q), {x, -Infinity, Infinity}, Method -> {"LobattoKronrodRule"}];
\end{verbatim}
The score ${\partial \log  f(x;\boldsymbol{\theta}) \over \partial \rho}$ for $\rho$ is given by
\begin{verbatim}
er := -\[Sigma]/Z + 1/(\[Rho] + \[Delta]*T[p, \[Mu], \[Sigma], \[Alpha]]);
\end{verbatim}
The score ${\partial \log  f(x;\boldsymbol{\theta}) \over \partial \delta}$ for $\delta$ is given by
\begin{verbatim}
ed := -\[Delta]/Z*(1 + NIntegrate[CDF[NormalDistribution[\[Mu], \[Sigma]],
-\[Alpha]*y^0.5]*PDF[GammaDistribution[p, 1], y], {y, 0, Infinity}, 
Method -> {"LobattoKronrodRule"}] - 
NIntegrate[CDF[NormalDistribution[\[Mu], \[Sigma]], \[Alpha]*y^0.5]*
PDF[GammaDistribution[p, 1], y], {y, 0, Infinity}, 
Method -> {"LobattoKronrodRule"}]) + 
T[p, \[Mu], \[Sigma], \[Alpha]] / 
(\[Rho] + \[Delta]*T[p, \[Mu], \[Sigma], \[Alpha]]);
\end{verbatim}
An element for ${\partial \log  f(x;\boldsymbol{\theta}) \over \partial \mu}$ of Fisher information matrix based on $\log_q$ is computed by the expression and its corresponding elements of matrix \cite{Cankaya2018}:
\begin{verbatim}
Imm := NIntegrate[em^2*f[\[Mu], \[Sigma], \[Alpha], \[Rho], \[Delta]]^(2 - q), 
{x, -Infinity, Infinity}, Method -> {"LobattoKronrodRule"}]
\end{verbatim}
The Fisher matrix M and inverse of M are given by
\begin{verbatim}
M = n*{{Imm, Ims, Ima, Imr, Imd}, {Ims, Iss, Isa, Isr, Isd}, {Ima, 
		Isa, Iaa, Iar, Iad}, {Imr, Isr, Iar, Irr, Ird}, {Imd, Isd, Iad, 
		Ird, Idd}};IM := Inverse[M]/n
\end{verbatim}

\end{appendices}

\end{document}